\documentclass[journal,twoside,web]{ieeecolor}
\usepackage{generic}
\usepackage{cite}
\usepackage{amsmath,amssymb,amsfonts}
\usepackage{algorithmic}
\usepackage{graphicx}
\usepackage{textcomp}
\usepackage{booktabs}
\newtheorem{theorem}{Theorem}
\newtheorem{lemma}{Lemma}
\newtheorem{assumption}{Assumption}
\newtheorem{definition}{Definition}
\newtheorem{proposition}{Proposition}
\newtheorem{remark}{Remark}
\newtheorem{corollary}{Corollary}
\usepackage{algorithm}
\usepackage{hyperref}
\hypersetup{hidelinks}
\usepackage[caption=false,font=footnotesize]{subfig}
\def\BibTeX{{\rm B\kern-.05em{\sc i\kern-.025em b}\kern-.08em
    T\kern-.1667em\lower.7ex\hbox{E}\kern-.125emX}}
\begin{document}
\title{Tube-Based Robust Data-Driven Predictive Control for Linear Time-Invariant Systems}
\author{Chi Wang, David Angeli, \IEEEmembership{Fellow, IEEE}
\thanks{Chi Wang and David Angeli are with the Control and Power Group, Department of Electrical and Electronic Engineering, Imperial College London, London SW7 2AZ, U.K. (e-mail: chi.wang23@imperial.ac.uk; d.angeli@imperial.ac.uk).}%
}

\maketitle
\thispagestyle{headings}

\begin{abstract}
This paper presents a tractable tube-based robust data-driven predictive control scheme that uses finite noisy input-state data from an unknown discrete-time linear time-invariant (LTI) system. A simplex constraint is imposed on the Hankel coefficient vector, yielding explicit polyhedral bounds on the prediction mismatch induced by bounded measurement noise. Using certified initial and terminal robust positively invariant (RPI) sets, we derive a tube-tightened formulation whose online optimization problem is a convex quadratic program (QP). The resulting controller guarantees recursive feasibility, robust satisfaction of input and state constraints, and practical input-to-state stability of the closed loop with respect to measurement noise. An input--output implementation is also developed, with sufficient conditions for retaining the closed-loop guarantees. Numerical examples illustrate the effectiveness, robustness, and closed-loop performance of the proposed method.
\end{abstract}


\section{Introduction}
\label{sec:introduction}
Data-driven control design has recently emerged as a powerful alternative to
classical model-based methods. Within the so-called data-informativity
framework, one seeks to synthesize controllers directly from a finite set of
input-state (or input-output) trajectories of an unknown linear time-invariant (LTI) system, without performing a separate identification step. This has led
to a variety of finite-data synthesis results for state-feedback regulation and
performance optimization, including data-driven stabilization~\cite{DePersisTesi2020TAC}, linear quadratic regulator (LQR) design~\cite{Dorfler2023TAC,Berberich2020ACC}, and state-feedback controllers
constructed directly from noisy data~\cite{bisoffi2021trade,vanWaardeMesbahi2020IFAC}.

While these works offer guarantees for unconstrained regulation and
performance optimization, the presence of input and state constraints remains a
major challenge in safety-critical applications. Model predictive control (MPC)
has become the standard paradigm for handling constraints in linear and
nonlinear systems by repeatedly solving a finite-horizon optimal control
problem online~\cite{AllgowerZheng2012,RawlingsMayneDiehl2017,Angeli2012EconomicMPC}.
Robust MPC schemes ensure constraint satisfaction and stability under bounded
disturbances; see, e.g.,~\cite{Kothare1996,Mayne2006ROF,Limon2002CDC}. Among these, tube-based MPC is a widely used paradigm: a nominal prediction is optimized subject to
tightened constraints, while a predesigned feedback law keeps the true state
within a robust invariant ``tube'' around the nominal trajectory~\cite{Rakovic2005MRPI,Mayne2005}. This separation yields simple online
optimization problems (often quadratic programs), preserves tunable prediction
horizons, and provides a transparent link between disturbance bounds and the
resulting conservatism. Motivated by these advantages, we also adopt a tube-based framework in this
paper, but in a fully data-driven setting. Instead of an explicit parametric
model, the tube dynamics and constraint tightenings are constructed directly
from measured trajectories by exploiting Willems' fundamental lemma~\cite{Willems2005}. This connects tube-based MPC with a complementary line
of research rooted in the behavioral approach and data-driven predictive
control.

Willems' fundamental lemma states that, under a suitable persistence-of-excitation
condition on the input, the columns of a Hankel matrix built from a single
measured trajectory span the set of all trajectories of the underlying LTI
system~\cite{Willems2005,Markovsky2008IJC}. This insight has enabled a
variety of behavioral and data-driven analysis and control methods, including
state-feedback design, dissipativity analysis, and predictive control
schemes that parametrize predicted trajectories directly in terms of data~\cite{DePersisTesi2020TAC,Romer2019LCSS,vanWaarde2020,berberich2020trajectory}.
In particular, data-enabled predictive control (DeePC)~\cite{coulson2019data}
and related data-driven MPC formulations~\cite{Huang2019CDC,Carlet2020ECCE} replace the explicit state-space model
by Hankel matrices of measured inputs and outputs, and optimize over
data-driven trajectories~\cite{Markovsky2021ARC}. Several extensions improve numerical robustness
and performance by introducing regularization, slack variables, or
distributionally robust costs, and some works address noise and decentralized
variants~\cite{Berberich2020IFAC,Huang2023TAC,Huang2022TCST,Coulson2019CDC}. However, most of
these contributions primarily enhance open-loop performance and numerical
robustness, while rigorous guarantees on recursive feasibility, closed-loop
constraint satisfaction, and stability in the presence of noisy data remain
largely unexplored.

To bridge this gap, a number of robust data-driven MPC schemes have been
proposed. The robust data-driven MPC formulations in~\cite{Berberich2021TAC,Berberich2020} provide some of the first closed-loop
theoretical guarantees under measurement noise. However, key tightening parameters are
computed under an assumption of noise-free data, and the resulting bounds tend
to be conservative. Other robust data-driven MPC schemes for linear systems,
such as~\cite{kerz2023data,Nguyen2023IFAC}, also rely on Hankel representations
but assume noise-free data in the design stage. The data-driven min--max MPC scheme in~\cite{xie2026} explicitly accounts for the worst-case effect of model uncertainty compatible with the data and achieves low conservatism; however, several design quantities cannot be computed directly from the data. A related work~\cite{Hu2025TAC} provides robust constraint-satisfaction
guarantees for systems with bounded disturbances and measurement noise,
and includes an input--output extension. It minimizes a cost upper bound
over a data-consistent model family, whereas the present formulation
directly parametrizes finite-horizon trajectories through Hankel coefficients.

In this paper, we develop a tube-based robust data-driven predictive control (TRDDPC) scheme that
combines the structural advantages of tube-based MPC with the flexibility of
Willems' fundamental lemma. We consider discrete-time constrained LTI systems
with additive measurement noise taking values in a known polyhedral set, and
assume that noisy input--state data are available offline. The
data-driven predictor is represented by Hankel matrices built from
an offline input--state trajectory, and a robust tube is constructed entirely in terms of the offline
data. We show that the
resulting controller guarantees recursive feasibility, robust closed-loop
satisfaction of all input and state constraints, and practical input-to-state
stability (ISS). At the
same time, the online optimization retains a simple computational structure of a quadratic
program with tunable prediction horizon, akin to classical robust tube-based
MPC, while using fewer decision variables and less conservative tightenings
than existing robust data-driven MPC schemes. Moreover, all design parameters can be computed directly from the noisy offline data.

The offline consistency-set construction, residual-bound certification, feedback synthesis, and robust positively invariant (RPI) set computations build on~\cite{WangAngeli2026CDCsub}; the model set provides a set-membership outer description of the unknown dynamics. The difficulty specific to the predictive formulation is that measurement noise corrupts the Hankel matrices themselves, so the prediction residual depends on both the unknown data errors and the optimized Hankel coefficients. Current measurement errors also affect prediction initialization and the implemented feedback.

The simplex constraint yields residual bounds valid for every feasible coefficient, enabling offline tube tightening. It also restricts predictions to a finite convex family: shifting a feasible trajectory and appending a stabilizing input need not remain representable by the same data. We address these coupled requirements through a tail-coverage condition that supplies an explicit simplex-preserving shift, a bound on the data-induced appended input deviation, and compatible error and terminal invariant sets. These ingredients allow the tube-MPC feasibility and Lyapunov arguments to be carried out for the noisy Hankel predictor.
For single-input single-output systems of known order, we also formulate an input--output implementation using measured histories and establish sufficient conditions for retaining these guarantees.

The remainder of this paper is organized as follows. Section~\ref{sec:Preliminaries}
states the required notation, the problem setup, and the resulting data-driven representation. In Section~\ref{sec:Robust} we formulate the tube-based robust data-driven predictive control
scheme, develop the tube construction and constraint-tightening procedure, and
establish the main results on recursive feasibility, closed-loop constraint
satisfaction and stability. Section~\ref{sec:io_implementation} presents an input--output
implementation. Section~\ref{sec:examples} illustrates the
proposed TRDDPC scheme on numerical examples,
and Section~\ref{sec:conclusions} concludes the paper.
\section{Preliminaries}\label{sec:Preliminaries}
This section collects notation and preliminaries. Section~\ref{sec:notation} fixes notation, Section~\ref{subsec:ProblemSetup} states the considered problem, Section~\ref{sec:available} specifies the assumptions on available data, and Section~\ref{sec:Data-driven} recalls the data-driven representation used throughout the paper.
\subsection{Notation}\label{sec:notation}
For $n\in\mathbb{Z}_{>0}$, $I_n$ is the $n\times n$ identity and $0$ a zero vector/matrix of compatible dimension.
For a matrix $A$, $A^\top$ is the transpose.
We write $\mathbb{N}:=\{1,2,\ldots\}$, $\mathbb{N}_0:=\{0,1,2,\ldots\}$, and for $a\le b$, $\mathbb{N}_{a:b}:=\{a,a{+}1,\ldots,b\}$.

For a vector $x$, $\|x\|_2$ and $\|x\|_\infty$ denote the Euclidean and $\ell_\infty$-norms.
For matrices, $\|\cdot\|_2$ denotes the induced Euclidean norm.
For sets $\mathcal A,\mathcal B\subset\mathbb{R}^n$, define the Minkowski sum and Pontryagin difference as
$\mathcal A\oplus\mathcal B := \{a+b \mid a\in\mathcal A,\ b\in\mathcal B\}$ and
$\mathcal A\ominus\mathcal B := \{x \in \mathbb{R}^n \mid \{x\}\oplus\mathcal B \subseteq \mathcal A\}$.
For a matrix $K$, $K\mathcal A:=\{Ka\mid a\in\mathcal A\}$ denotes the linear image.
The convex hull is $\operatorname{co}(\mathcal S)$ and the Kronecker product is $\otimes$.
In predictive control, $x_{\ell|k}$ denotes the predicted state $\ell$ steps ahead at time $k$, and $u_{\ell|k}$ the corresponding predicted input.
For a vector sequence, let
$\operatorname{col}(s_i,\ldots,s_j):=[s_i^\top,\ldots,s_j^\top]^\top$.
We write $s_{[i:j]}:=(s_i,\ldots,s_j)$ for the finite sequence and
$s_{[i,j]}:=\operatorname{col}(s_i,\ldots,s_j)$ for its stacked column
vector. For predictions at time $k$,
$s_{[i|k,j|k]}:=\operatorname{col}(s_{i|k},\ldots,s_{j|k})$.

For a length-$T$ sequence $s_{[0{:}T{-}1]}$ with $s_i\in\mathbb{R}^p$ and depth $L{+}1$, the block Hankel matrix is
\[
H_{L+1}\!\big(s_{[0{:}T{-}1]}\big):=
\begin{bmatrix}
s_0 & s_1 & \cdots & s_{T-L-1}\\
s_1 & s_2 & \cdots & s_{T-L}\\
\vdots & \vdots & \ddots & \vdots\\
s_L & s_{L+1} & \cdots & s_{T-1}
\end{bmatrix}.
\]
We write $H_s:=H_{L+1}(s_{[0{:}T{-}1]})$.
For $\ell\in\mathbb{N}_{0:L}$, $H_{s_{[\ell]}}$ denotes the $\ell$-th $p$-row block of $H_s$, and for $i\le j$,
$H_{s_{[i:j]}}$ is the block-row submatrix containing blocks $i,\ldots,j$.

\subsection{Problem Setup}\label{subsec:ProblemSetup}
We consider an LTI system
\begin{subequations}\label{eq:system_model}
\begin{align}
    x_{k+1} &= A^\ast x_k + B^\ast u_k, \label{eq:lti_system} \\
    \hat{x}_k &= x_k + w_k, \label{eq:measurement_model}
\end{align}
\end{subequations}
where $x_k\in\mathbb{R}^n$ and $u_k\in\mathbb{R}^m$ denote the state and input, respectively, and
$\hat{x}_k\in\mathbb{R}^n$ is the measured state.
The term $w_k\in\mathbb{R}^n$ denotes the measurement noise at time $k$ and satisfies $w_k\in\mathcal{W}$ for all $k$ (see Assumption~\ref{as:measurement noise}).
The true (unknown but fixed) system matrices are $A^\ast\in\mathbb{R}^{n\times n}$ and $B^\ast\in\mathbb{R}^{n\times m}$, and the pair $(A^\ast,B^\ast)$ is controllable.
The state dimension $n$ is known.
This paper considers measurement noise without independent process disturbances; treating both jointly is left for future work.

Our objective is to regulate the system to its equilibrium using the
available noisy state measurements, while satisfying the closed-loop
constraints $u_k\in\mathcal U$ and $x_k\in\mathcal X$.
Both $\mathcal U$ and $\mathcal X$ are user-specified convex polytopic sets:
\begin{subequations}\label{eq:constraint_sets}
\begin{equation}
\mathcal{U} = \left\{ u \in \mathbb{R}^{m} \;\middle|\; G_u u \leq h_u \right\},
\label{eq:constraint_sets_U}
\end{equation}
\begin{equation}
\mathcal{X} = \left\{ x \in \mathbb{R}^{n} \;\middle|\; G_x x \leq h_x \right\},
\label{eq:constraint_sets_X}
\end{equation}
\end{subequations}
where $G_u,G_x$ and $h_u,h_x$ have compatible dimensions.
We assume that both $\mathcal{U}$ and $\mathcal{X}$ contain the origin in their
interiors.
Such polyhedral constraints are standard in tube-based robust MPC formulations~\cite{Langson2004}.

\subsection{Available Data}\label{sec:available}
\begin{definition}
Let $L,T\in\mathbb{Z}_{>0}$ with $T\ge L$.
The input $u^d_{[0{:}T{-}1]}$ is said to be persistently exciting (PE) of order $L+1$ if
$\operatorname{rank}\!\big(H_{L+1}\!\big(u^d_{[0:T-1]}\big)\big)=m(L+1)$.
\end{definition}

\begin{assumption}\label{as:available-trajectory}
An offline input--state trajectory $\big(u^d_{[0{:}T{-}1]},\,\hat{x}^d_{[0{:}T{-}1]}\big)$ of length $T$
generated by system~\eqref{eq:system_model} is available.
Fix a depth $L{+}1$ with $1\le L<T$.
The input $u^d_{[0{:}T{-}1]}$ is persistently exciting of order $L{+}n{+}1$.
\end{assumption}

Given Assumption~\ref{as:available-trajectory}, define
$H_u := H_{L+1}(u^d_{[0{:}T{-}1]})$ and
$\hat H_x := H_{L+1}(\hat x^d_{[0{:}T{-}1]})$.
For the unmeasured clean state and noise, define
$H_x := H_{L+1}(x^d_{[0{:}T{-}1]})$ and
$H_w := H_{L+1}(w^d_{[0{:}T{-}1]})$, so that
$\hat H_x = H_x + H_w$.
Only $H_u$ and $\hat H_x$ are available; $H_x$ and $H_w$ are unknown.

\begin{assumption}\label{as:measurement noise}
The measurement noise $w_k\in\mathbb{R}^{n}$ satisfies $w_k\in\mathcal{W}$ for all $k$, where
\begin{equation}
\mathcal{W} := \{\, w \in \mathbb{R}^{n} \mid G_w w \le h_w \,\},
\label{eq:disturbance_polytope}
\end{equation}
is a known, compact, convex polytope with $0\in\operatorname{int}\mathcal W$.
\end{assumption}

\subsection{Data-driven Representation}\label{sec:Data-driven}
Willems' lemma provides an appealing data-driven representation of unknown LTI systems using trajectory data generated by a persistently exciting input signal~\cite{Willems2005}. In this paper, we state and use the corresponding Fundamental Lemma in the
input-state space; cf.~\cite{vanWaarde2020}.

\begin{lemma}\label{lem:FL_input_state}
Let $L,T\in\mathbb{Z}_{>0}$ with $T\ge L$.
Consider system~\eqref{eq:lti_system} with $(A^\ast, B^\ast)$ controllable. If the input sequence
$u^d_{[0{:}T{-}1]}$ is persistently exciting of order $L+n+1$, then any
$(L+1)$-long input--state sequence $({u}_{[k,k+L]},\,{x}_{[k,k+L]})$
is a valid trajectory of the system if and only if there exists
$g\in\mathbb{R}^{T-L}$ such that
\begin{equation}
\label{eq:data_eq_stack}
\begin{bmatrix}
{u}_{[k,k+L]}\\[.3em]
{x}_{[k,k+L]}
\end{bmatrix}
=
\begin{bmatrix}
H_u\\[.3em]
H_x
\end{bmatrix}g.
\end{equation}
\end{lemma}
In data-driven predictive control, the stacked Hankel matrices in~\eqref{eq:data_eq_stack} replace the explicit state-space model, and the
coefficient vector $g$ is treated as a decision variable that parametrizes the
predicted input-state trajectory. In the noise-free setting, Lemma~\ref{lem:FL_input_state} yields an exact behavioral characterization via Hankel-matrix equalities.

\begin{remark}
A necessary data-length condition for PE of order $L+n+1$ is
$T \ge (m+1)(L+n+1)-1$.
\end{remark}

Since the system matrices $(A^\ast,B^\ast)$ are not available, we define equilibria directly in terms of input-state trajectories.

\begin{definition}\label{def:equilibrium}
A pair $(u_s,x_s)\in\mathbb R^{m}\times\mathbb R^{n}$ is an equilibrium of the LTI system~\eqref{eq:lti_system}
if the constant sequence $(u_k,x_k)\equiv (u_s,x_s)$ for all $k\in\mathbb N_0$ is a trajectory of~\eqref{eq:lti_system}.
\end{definition}

Without loss of generality, we take the origin $(u_s,x_s) = (0,0)$, with
$0\in\mathbb R^m$ and $0\in\mathbb R^n$, as the equilibrium operating point
for simplicity; extensions to nonzero equilibria follow by a standard change
of variables. In the remainder of the paper, we develop a
robust data-driven predictive control scheme based on
Lemma~\ref{lem:FL_input_state}, operating in the presence of measurement noise
and using equilibria in the sense of Definition~\ref{def:equilibrium} for data collection and controller design.

\section{Tube-based robust Data-Driven Predictive Control}\label{sec:Robust}
In this section, we introduce a tube-based robust data-driven predictive control scheme for inexact data conditions. To ensure that the closed-loop (actual) state respects the original constraints in the presence of measurement noise, the scheme employs a data-driven constraint-tightening step. Section~\ref{Scheme} presents the tube-based robust data-driven predictive control scheme, while
Section~\ref{AA} details the tightening procedure.
Recursive feasibility, closed-loop constraint satisfaction, and practical ISS are
established in Sections~\ref{sec:rec-feas}, \ref{Constraint Satisfaction},
and~\ref{ISS}, respectively.

\subsection{Proposed Data-Driven Predictive Control Scheme}\label{Scheme}
To mitigate noise amplification in the data-driven predictions, we constrain the Hankel coefficient vector $g$ to the unit simplex,
\begin{equation}\label{eq:simplex_regularizer}
g\in\Delta^{T-L},\quad
\Delta^{T-L}:=\{g\in\mathbb{R}^{T-L}\mid g\ge0,\,\mathbf{1}^{\top} g=1\},
\end{equation}
where $\mathbf{1}\in\mathbb{R}^{T-L}$ denotes the all-ones vector.
The entries of $g$ are nonnegative weights summing to one, so multiplication
by a Hankel matrix forms a convex combination of its columns. Each column
is an entire data window, and the same weights are used at every prediction
stage.

For the actual, unknown noise Hankel matrix $H_w$ associated with the
offline data, each block $H_{w_{[\ell]}}g$ is therefore a convex combination
of noise samples in $\mathcal W$. Hence,
\begin{equation}\label{eq:block_bound_poly}
H_{w_{[\ell]}}g\in\mathcal W,\qquad \ell\in\mathbb N_{0:L}.
\end{equation}
This provides a uniform blockwise noise bound for every feasible coefficient
and prevents amplification of the offline noise by the Hankel coefficients.
However, the simplex restriction may reduce the representable trajectory set
and the controller's feasible region. Its effect on trajectory representation
and feasibility is quantified in Section~\ref{sec:scalar_conservatism}.

Let $K$ be a fixed feedback gain satisfying the design conditions in
Section~\ref{sec:rec-feas}, and define $\tilde K:=I_{L+1}\otimes K$.
The predictor is formed directly from the offline inputs and noisy state
measurements using $H_u-\tilde K\hat H_x$ and $\hat H_x$.
For $g\in\Delta^{T-L}$ and $\ell\in\mathbb N_{0:L}$, define
\[
z_{\ell|k}:=\hat H_{x_{[\ell]}}g,\qquad
v_{\ell|k}:=(H_{u_{[\ell]}}-K\hat H_{x_{[\ell]}})g.
\]
The predicted input deviation $v_{\ell|k}$ defines the planned input
\begin{equation}\label{eq:planned_input}
u^{\mathrm p}_{\ell|k}:=v_{\ell|k}+Kz_{\ell|k}=H_{u_{[\ell]}}g.
\end{equation}
This is the input appearing in the stage cost and tightened constraints.
Measurement noise affects both the data-driven predictor and its online
initialization. For the actual, unknown noise Hankel matrix $H_w$ associated
with the offline data, the predictor satisfies
\begin{equation}\label{eq:offline_noise_identity}
z_{\ell+1|k}-A^\ast z_{\ell|k}-B^\ast u_{\ell|k}^{\mathrm p} =\bigl(H_{w_{[\ell+1]}}-A^\ast H_{w_{[\ell]}}\bigr)g,
\end{equation}
for $\ell\in\mathbb N_{0:L-1}$. Thus the offline data error enters the
prediction through its product with the decision variable $g$.
Separately,
\begin{equation}\label{eq:initial_noise_identity}
z_{0|k}-x_k=(z_{0|k}-\hat x_k)+w_k.
\end{equation}
The simplex constraint bounds the aggregated offline noise, while the
initial tube constraint bounds the prediction--measurement mismatch at
initialization.

Given the current measurement $\hat x_k$, we solve the following robust
data-driven predictive optimal control problem (OCP):
\begin{subequations}\label{eq:robust_deepc}
\begin{align}
\underset{\{Z_k,\,V_k,\,g\}}{\text{minimize}}\quad
& \sum_{\ell=0}^{L-1} J_{\mathrm{s}}\!\bigl(z_{\ell|k}, v_{\ell|k}+K z_{\ell|k}\bigr)
  \;+\; J_{\mathrm{f}}\!\bigl(z_{L|k}\bigr),
\label{eq:robust_deepc_obj}\\[1mm]
\text{subject to}\quad
& z_{0|k}-\hat x_k \;\in\; \mathcal E,
\label{eq:robust_deepc_init}\\[1mm]
& \begin{bmatrix}
  v_{[0|k,\,L|k]}\\[.3mm]
  z_{[0|k,\,L|k]}
  \end{bmatrix}
  =
  \begin{bmatrix}
  H_u - \tilde K\,\hat H_x\\[.3mm]
  \hat H_x
  \end{bmatrix} g,
\label{eq:robust_deepc_dyn}\\[1mm]
& g\in\Delta^{T-L},
\label{eq:robust_deepc_simplex}\\[1mm]
& z_{\ell|k} \;\in\; \mathcal Z,\qquad \forall\, \ell \in \mathbb{N}_{1:L},
\label{eq:robust_deepc_state}\\[1mm]
& v_{\ell|k}+K z_{\ell|k} \;\in\; \hat{\mathcal U},\qquad \forall\, \ell \in \mathbb{N}_{0:L-1},
\label{eq:robust_deepc_input}\\[1mm]
& z_{L|k} \;\in\; \mathcal Z_{\mathrm f}.
\label{eq:robust_deepc_terminal}
\end{align}
\end{subequations}
Here, $V_k:=v_{[0|k,\,L-1|k]}$ and $Z_k:=z_{[0|k,\,L|k]}$ denote the predicted input-deviation and state sequences, respectively; $\mathcal Z\subset\mathbb R^n$ and $\hat{\mathcal U}\subset\mathbb R^m$ are tightened state and input polytopes, and $\mathcal E$ and $\mathcal Z_f\subseteq \mathcal Z$ are chosen as RPI polytopes in Sections~\ref{sec:mismatch_bounds} and~\ref{sec:terminal_design}, respectively.
The tightened state and input constraints~\eqref{eq:robust_deepc_state}--\eqref{eq:robust_deepc_input}
guarantee closed-loop constraint satisfaction in the presence of measurement noise.
The initial constraint~\eqref{eq:robust_deepc_init} and the terminal constraint~\eqref{eq:robust_deepc_terminal} are pivotal for recursive feasibility and stability. All constraints in~\eqref{eq:robust_deepc} are affine or, by design, polyhedral. The stage and terminal costs are
quadratic:
\begin{equation}\label{eq:QR_costs}
\begin{aligned}
&J_{\mathrm s}\!\big(z_{\ell|k},\,v_{\ell|k}{+}Kz_{\ell|k}\big)
= z_{\ell|k}^\top Q z_{\ell|k}\\
&\qquad+ (v_{\ell|k}{+}Kz_{\ell|k})^\top R (v_{\ell|k}{+}Kz_{\ell|k}),\\
&J_{\mathrm f}(z_{L|k})= z_{L|k}^\top P_{L} z_{L|k},
\end{aligned}
\end{equation}
with $Q\succ0$ and $R\succ0$. The terminal weight $P_{L}\succ0$ is selected
later, consistently with the stabilizing gain $K$, to satisfy the terminal
Lyapunov decrease condition established in Section~\ref{ISS}.
The optimization problem is a convex quadratic program (QP) whose objective is
strictly convex in $(Z_k,V_k)$.
Whenever feasible, it therefore has a unique optimal pair
$(Z_k^\star,V_k^\star)$ and a unique implemented input for each fixed
$\hat x_k$.

The Hankel coefficient $g^\star$ need not be unique. The terminal deviation
$v_{L|k}^\star$, which is excluded from the objective, is associated with
the selected $g^\star$ and need not be unique either.
For any selected optimizer, the associated terminal deviation is
determined by the last block row of the Hankel parametrization.

The implemented control law is
\begin{equation}\label{eq:kappa_deepc_clean}
u_k=\kappa_L^\star(\hat x_k) :=v_{0|k}^\star+K\hat x_k =u_{0|k}^{\mathrm p,\star}+K(\hat x_k-z_{0|k}^\star).
\end{equation}
Thus, the feedback correction uses the current noisy measurement and the
optimized nominal state.

\begin{remark}\label{rem:additional_regularization}
An additional quadratic penalty $\lambda_g\|g\|_2^2$, $\lambda_g>0$, can
complement the simplex constraint by influencing trajectory selection and
removing coefficient nonuniqueness. Noise-mitigation and robust-optimization
interpretations of such penalties are available in related
formulations~\cite{Berberich2021TAC,HuangQuadratic2021}.
The feasible set is unchanged, and an improvement in prediction accuracy
is not guaranteed under the bounded-noise assumptions used here.
The present analysis retains the trajectory-based objective; an additional
penalty requires a corresponding value-function decrease analysis.
\end{remark}
We define the admissible set of predicted initial states as
\begin{equation}\label{eq:Z_L_feas}
\begin{aligned}
\mathcal Z_{L}^{\mathrm{anc}}
&:=
\Big\{
z\in\mathbb R^{n}\ \Big|\ 
\exists(Z_k,V_k,g):\\
&\qquad \eqref{eq:robust_deepc_dyn}\text{--}\eqref{eq:robust_deepc_terminal}
\ \text{hold}, z_{0|k}=z
\Big\}.
\end{aligned}
\end{equation}
The corresponding admissible set of measured states is
\begin{equation}\label{eq:Xhat_L_feas_def}
\begin{aligned}
\hat{\mathcal X}_L^{\mathrm{feas}}
&:=\bigl\{\hat x\in\mathbb R^n\mid
\exists z\in\mathcal Z_{L}^{\mathrm{anc}}\ \text{with}\ z-\hat x\in\mathcal E\bigr\}\\
&=\mathcal Z_{L}^{\mathrm{anc}}\oplus(-\mathcal E).
\end{aligned}
\end{equation}
Finally, we define the
closed-loop feasible set (data-driven region of attraction) as
\begin{equation}\label{eq:X_L_feas_def}
\mathcal X_L^{\mathrm{feas}}
:=\bigl\{x\in\mathbb R^n\mid x+\mathcal W\subseteq\hat{\mathcal X}_L^{\mathrm{feas}}\bigr\}
=\hat{\mathcal X}_L^{\mathrm{feas}}\ominus\mathcal W.
\end{equation}
Thus, if \(x_0\in\mathcal X_L^{\mathrm{feas}}\), then the initial
measurement satisfies \(\hat x_0=x_0+w_0\in\hat{\mathcal X}_L^{\mathrm{feas}}\)
for all \(w_0\in\mathcal W\).
In view of Theorem~\ref{thm:rec_feas_shift}, the measurement sequence
$\hat x_k$ then remains in \(\hat{\mathcal X}_L^{\mathrm{feas}}\) and
the OCP~\eqref{eq:robust_deepc} admits a feasible solution at every sampling
instant.

\paragraph*{Offline design and online implementation}\label{sec:offline_order}
Given the noisy offline data, offline design verifies
the data conditions and constructs the error tube, tightened constraints,
and terminal ingredients using a certified fixed gain $K$, residual bound,
and data-consistent model family. Their invariance and
state/input compatibility are checked before implementation, together with
feasibility of the initial OCP for the available measurement.

Online, each step requires a state measurement, solution of the
QP~\eqref{eq:robust_deepc}, and application of
$u_k=v_{0|k}^\star+K\hat x_k$.

The resulting robust data-driven predictive control scheme is implemented in a
receding-horizon fashion, as summarized in Algorithm~\ref{alg:TRDDPC}.

\begin{algorithm}[t]
  \caption{Tube-based robust data-driven predictive control (TRDDPC)}
  \label{alg:TRDDPC}
  \begin{algorithmic}[1]
    \STATE Measure the current noisy state $\hat x_k$ at time step $k$.
    \STATE Solve the OCP~\eqref{eq:robust_deepc}, obtaining an optimal solution
           $(Z_k^\star,V_k^\star,g^\star)$.
    \STATE Apply the input
           $
             u_k \;=\; v_{0|k}^\star + K\,\hat x_k.
           $
    \STATE Set $k \gets k{+}1$ and return to Step~1.
  \end{algorithmic}
\end{algorithm}

\subsection{Constraint Tightening}\label{AA}
We derive the mismatch bounds and error tube for a fixed gain $K$
certified as in Section~\ref{sec:gain_certificate}, then use them to
construct the tightened state and input sets.

\subsubsection{Prediction and Auxiliary Signals}\label{sec:prediction_signals}
To tighten the constraints, we compare the Hankel prediction with the plant
evolution under any fixed feasible solution of the OCP at time $k$.
For this analysis, the feedforward sequence $v_{[0|k,L-1|k]}$ is held
fixed, and $z_{\ell|k}$ denotes the state prediction of that solution.
Set $x_{0|k}=x_k$ and define
\begin{equation}\label{eq:meas_def_TAC}
\hat x_{\ell|k}:=x_{\ell|k}+w_{k+\ell},\qquad w_{k+\ell}\in\mathcal W,
\end{equation}
for $\ell\in\mathbb N_{0:L}$, where $w_{k+\ell}$ denotes the
measurement-noise sample at time $k+\ell$. The plant rollout satisfies
\begin{equation}\label{eq:x_dynamics_main}
x_{\ell+1|k}
=
A^\ast x_{\ell|k}+B^\ast\big(v_{\ell|k}+K\hat x_{\ell|k}\big).
\end{equation}
These dynamics hold for $\ell\in\mathbb N_{0:L-1}$. The future rollout is introduced for
analysis; only $\hat x_{0|k}=\hat x_k$ is available to the OCP at time $k$.
Since only the first input is applied, $x_{k+1}=x_{1|k}$. In the actual
receding-horizon implementation, the feedforward input is re-optimized at
each time step; therefore, the resulting state at time $k+\ell$ need not
coincide with the fixed-plan rollout $x_{\ell|k}$ for $\ell\ge2$.

Define two auxiliary one-step model evaluations, for
$\ell\in\mathbb N_{0:L-1}$:
\begin{equation}\label{eq:tildex_dynamics_main}
\tilde x_{\ell+1|k}
=
A^\ast \hat x_{\ell|k}+B^\ast\big(v_{\ell|k}+K\hat x_{\ell|k}\big),
\qquad
\tilde x_{0|k}=\hat x_k,
\end{equation}
\begin{equation}\label{eq:barx_dynamics_main}
\bar x_{\ell+1|k}
=
A^\ast z_{\ell|k}+B^\ast\big(v_{\ell|k}+K z_{\ell|k}\big),
\qquad
\bar x_{0|k}=z_{0|k},
\end{equation}
where the model is evaluated from $\hat x_{\ell|k}$ and $z_{\ell|k}$,
respectively. Table~\ref{tab:signal_roles} summarizes the five signals.
\begin{table}[t]
\caption{Signals used in the constraint-tightening analysis}
\label{tab:signal_roles}
\centering\footnotesize
\begin{tabular}{@{}p{.18\columnwidth}p{.76\columnwidth}@{}}
\toprule
Symbol & Meaning\\
\midrule
$z_{\ell|k}$ & State prediction of the fixed Hankel-based OCP solution.\\[.3em]
$x_{\ell|k}$ & Plant rollout under the fixed feedforward plan and measured-state feedback.\\[.3em]
$\hat x_{\ell|k}$ & Measurement of that rollout; only $\hat x_{0|k}=\hat x_k$ is available at time $k$.\\[.3em]
$\tilde x_{\ell|k}$ & Auxiliary one-step model evaluation from the measured rollout state.\\[.3em]
$\bar x_{\ell|k}$ & Auxiliary one-step model evaluation from the Hankel-predicted state.\\
\bottomrule
\end{tabular}
\end{table}

With these definitions, for $\ell\in\mathbb N_{0:L}$,
\begin{equation}\label{eq:tube_chain_intro}
\begin{aligned}
z_{\ell|k}-x_{\ell|k}
&=
\underbrace{(z_{\ell|k}-\bar x_{\ell|k})}_{\text{offline mismatch}}
+\underbrace{(\bar x_{\ell|k}-\tilde x_{\ell|k})}_{\text{propagated mismatch}}\\
&\quad+\underbrace{(\tilde x_{\ell|k}-x_{\ell|k})}_{\text{measurement noise}}.
\end{aligned}
\end{equation}
The auxiliary states are one-step model evaluations. They separate the
offline-data residual and the one-step measurement-noise contribution
from the error propagated by the closed-loop dynamics. This permits
the propagated error to be bounded by a single invariant tube, which
also accommodates re-anchoring at the next sampling instant.

\subsubsection{Data-Consistent Residual Bounds}\label{sec:residual_bounds}
To make the bounds computable from data, we derive a one-step residual
representation of the Hankel prediction for each model consistent with
the measured data in the corresponding window.

We start with a regression model for the measured state. Under
Assumption~\ref{as:measurement noise}, with $\hat x_k = x_k + w_k$ and
$w_k \in \mathcal W$, the measured state satisfies
\begin{equation}\label{eq:regression_TAC_poly}
\hat x_{k+1}
=
A^\ast\,\hat x_k
+
B^\ast\,u_k
+
d_k,
\qquad
d_k := w_{k+1} - A^\ast w_k.
\end{equation}
Define the Minkowski gauge and the directed induced gain by
\begin{equation}\label{eq:directed_gain_def}
\begin{aligned}
p_{\mathcal W}(x)&:=\inf\{\alpha\ge0\mid x\in\alpha\mathcal W\},\\
\mu_{\mathcal W}(A)&:=\max_{w\in\mathcal W}p_{\mathcal W}(-Aw).
\end{aligned}
\end{equation}
Since $-A\mathcal W\subseteq\mu_{\mathcal W}(A)\mathcal W$, for any
$w_k,w_{k+1}\in\mathcal W$,
\begin{equation}\label{eq:d_norm_bound_poly}
w_{k+1}-Aw_k\in(1+\mu_{\mathcal W}(A))\mathcal W.
\end{equation}
Thus, whenever $\gamma\ge\mu_{\mathcal W}(A)$, the residual belongs to
\begin{equation}\label{eq:D_gamma_def}
\mathcal D(\gamma):=(1+\gamma)\mathcal W
=\{d\in\mathbb R^n\mid G_wd\le(1+\gamma)h_w\}.
\end{equation}

We now specialize~\eqref{eq:regression_TAC_poly} to the offline data
$\big(u^d_{[0{:}T{-}1]},\,\hat x^d_{[0{:}T{-}1]}\big)$ from
Assumption~\ref{as:available-trajectory}. For each $i\in\mathbb{N}_{0:T-2}$, define the polyhedral one-step consistency set
\begingroup
\thinmuskip=2mu\medmuskip=2mu\thickmuskip=2mu\relax
\begin{IEEEeqnarray}{l}\label{eq:C_i_poly_def}
\mathcal C_i^{\mathrm{poly}}:=\{(A,B)\mid G_w(\hat x^d_{i+1}-A\hat x^d_i-Bu^d_i)\le(1+\gamma)h_w\}.\IEEEeqnarraynumspace
\end{IEEEeqnarray}
\endgroup

Because the prediction horizon has length $L{+}1$, we only require consistency
over $L{+}1$-long windows aligned with the $\ell$-th block-rows of the Hankel matrices.
For each $\ell\in\mathbb{N}_{0:L-1}$, define the window-consistent model set
\begin{equation}\label{eq:I_set_poly_window}
\mathcal I_\ell^{\mathrm{poly}}
:=
\bigcap_{i=\ell}^{T-L-1+\ell} \mathcal C_i^{\mathrm{poly}}.
\end{equation}

Let $\gamma^\ast\ge\mu_{\mathcal W}(A^\ast)$ be a bound obtained by
successful application of the data-driven certification procedure
in~\cite[Appendix]{WangAngeli2026CDCsub} to the same offline data.

With the certified bound $\gamma^\ast$ fixed, define, for $i=0,\ldots,T-2$,
\begingroup
\thinmuskip=2mu\medmuskip=1mu\thickmuskip=1mu\relax
\begin{IEEEeqnarray}{l}\label{eq:C_i_star_TAC}
\mathcal C_i^{\mathrm{poly},\ast}:=\{(A,B)\mid G_w(\hat x^d_{i+1}-A\hat x^d_i-Bu^d_i)\le(1+\gamma^\ast)h_w\}.\IEEEeqnarraynumspace
\end{IEEEeqnarray}
\endgroup
The full data-consistency set is
\begin{equation}\label{eq:I_full_TAC}
\mathcal I_{\mathrm{full}}^{\mathrm{poly},\ast}
:=\bigcap_{i=0}^{T-2}\mathcal C_i^{\mathrm{poly},\ast}.
\end{equation}

\begin{assumption}\label{as:fullrowrank}
Let $X_0=[\hat x^d_0\ \cdots\ \hat x^d_{T-2}]$, and
$U_0=[u^d_0\ \cdots\ u^d_{T-2}]$. We assume that the stacked data matrix
\[
Z_0^{\mathrm d}=\begin{bmatrix} X_0\\ U_0 \end{bmatrix}
\in \mathbb{R}^{(n+m)\times (T-1)}
\]
has full row rank, i.e.,
\(
\operatorname{rank}\big(Z_0^{\mathrm d}\big)=n+m.
\)
\end{assumption}

Since the certified bound $\gamma^\ast$ obtained using the procedure
in~\cite{WangAngeli2026CDCsub} ensures that
$(A^\ast,B^\ast)\in\mathcal I_{\mathrm{full}}^{\mathrm{poly},\ast}$,
Assumption~\ref{as:fullrowrank} implies that
$\mathcal I_{\mathrm{full}}^{\mathrm{poly},\ast}$ is a nonempty compact
polytope by~\cite[Lemma~2]{WangAngeli2026CDCsub}.
The same certification also gives the uniform bound
\begin{equation}\label{eq:uniform_directed_gain}
\max_{(A,B)\in\mathcal I_{\mathrm{full}}^{\mathrm{poly},\ast}}
\mu_{\mathcal W}(A)\le\gamma^\ast.
\end{equation}
For a fixed gain $K$, define the closed-loop matrix family
\begin{equation}\label{eq:A_K_poly_def}
{\mathcal A}_K
:= \bigl\{\, A{+}BK \ \big|\ (A,B)\in \mathcal I_{\mathrm{full}}^{\mathrm{poly},\ast}\,\bigr\},
\end{equation}
which is a polytope in the matrix space by linearity of the map
$(A,B)\mapsto A{+}BK$. 
For a set $\mathcal S\subseteq\mathbb R^n$, write
\begin{equation}\label{eq:robustimage}
\mathcal A_K\mathcal S
:=\{(A+BK)s\mid A+BK\in\mathcal A_K,\ s\in\mathcal S\}.
\end{equation}
The following proposition gives the required residual representation
directly from the data-consistency inequalities.
\begin{proposition}\label{prop:window_model_realization_TAC_poly}
Fix $\gamma\ge0$ and $\ell\in\mathbb N_{0:L-1}$, and let
$\mathcal I_\ell^{\mathrm{poly}}$ be defined using this value of $\gamma$.
For $(A,B)\in\mathcal I_\ell^{\mathrm{poly}}$, define
\begin{equation}\label{eq:residual_matrix}
R_\ell(A,B):=\hat H_{x_{[\ell+1]}}
-A\hat H_{x_{[\ell]}}-BH_{u_{[\ell]}}.
\end{equation}
Let $g\in\Delta^{T-L}$ satisfy the Hankel parametrization~\eqref{eq:robust_deepc_dyn}. Then
\begin{equation}\label{eq:z_dyn_window_TAC_poly}
z_{\ell+1|k}=Az_{\ell|k}+Bu_{\ell|k}^{\mathrm p}
+\Delta_\ell(A,B;g),
\end{equation}
where $u_{\ell|k}^{\mathrm p}=v_{\ell|k}+Kz_{\ell|k}$ and
\begin{equation}\label{eq:Delta_def_TAC_poly}
\Delta_\ell(A,B;g):=R_\ell(A,B)g.
\end{equation}
Moreover,
\begin{equation}\label{eq:Delta_in_Dgamma}
\Delta_\ell(A,B;g)\in\mathcal D(\gamma)=(1+\gamma)\mathcal W.
\end{equation}
\end{proposition}
\begin{proof}
By the definition of $\mathcal I_\ell^{\mathrm{poly}}$, each column of
$R_\ell(A,B)$ belongs to $\mathcal D(\gamma)$. Since $g\ge0$ and
$\mathbf1^\top g=1$, the product $R_\ell(A,B)g$ is a convex combination
of these columns and belongs to the same convex set. The identity follows
from $z_{\ell|k}=\hat H_{x_{[\ell]}}g$,
$z_{\ell+1|k}=\hat H_{x_{[\ell+1]}}g$, and
$u_{\ell|k}^{\mathrm p}=H_{u_{[\ell]}}g$.
\end{proof}

\subsubsection{Error Tube and Mismatch Bounds}\label{sec:mismatch_bounds}
We first bound the measurement-noise contribution and offline mismatch.
These bounds determine the additive term in the error dynamics; an
RPI error tube then bounds the propagated mismatch.

\paragraph{Measurement-noise contribution}
Subtracting~\eqref{eq:x_dynamics_main} from~\eqref{eq:tildex_dynamics_main} yields, for $\ell\in\mathbb N_{0:L-1}$,
\begin{equation}\label{eq:tilde_minus_x_TAC}
\tilde x_{\ell+1|k}-x_{\ell+1|k}
=
A^\ast(\hat x_{\ell|k}-x_{\ell|k})
=
A^\ast w_{k+\ell}.
\end{equation}
At the initial stage, instead,
$\tilde x_{0|k}-x_{0|k}=\hat x_k-x_k=w_k$.
Thus $\tilde x_{\ell|k}-x_{\ell|k}\in A^\ast\mathcal W$ for
$\ell\in\mathbb N_{1:L}$.
With the certified model set fixed, choose a compact convex polytope
$\mathcal W_A$ satisfying
\begin{equation}\label{eq:WA_uniform_LTI}
A\mathcal W\subseteq\mathcal W_A\subseteq\gamma^\ast(-\mathcal W), \qquad\forall\,(A,B)\in\mathcal I_{\mathrm{full}}^{\mathrm{poly},\ast}.
\end{equation}
The scalar choice $\mathcal W_A=\gamma^\ast(-\mathcal W)$ is admissible
by~\eqref{eq:uniform_directed_gain}. One may instead use the convex hull
of the images $A\mathcal W$ at the vertices of the model polytope,
or a polytopic outer approximation satisfying~\eqref{eq:WA_uniform_LTI}.
Then~\eqref{eq:tilde_minus_x_TAC} gives, uniformly over all fixed feasible plans,
\begin{equation}\label{eq:tilde_minus_x_poly_bound_short}
\tilde x_{\ell|k}-x_{\ell|k}
\in\mathcal W_A
\;=:\;
\mathcal E_w,
\qquad
\forall\,\ell\in\mathbb N_{1:L}.
\end{equation}

\paragraph{Offline mismatch}
Set $\gamma=\gamma^\ast$. The certified true-model inclusion gives
$(A^\ast,B^\ast)\in\mathcal I_{\mathrm{full}}^{\mathrm{poly},\ast}
\subseteq\mathcal I_\ell^{\mathrm{poly}}$ for every
$\ell\in\mathbb N_{0:L-1}$. Specializing
Proposition~\ref{prop:window_model_realization_TAC_poly} to the true pair
and writing $\Delta_\ell(g):=\Delta_\ell(A^\ast,B^\ast;g)$ yields
\begin{equation}\label{eq:z_true_update_step3_short}
z_{\ell+1|k}
=
A^\ast z_{\ell|k}
+
B^\ast u_{\ell|k}^{\mathrm p}
+
\Delta_\ell(g),
\end{equation}
where the equations for the actual offline data and $\hat H_x=H_x+H_w$ give
\begin{equation}\label{eq:Delta_true_step3_short}
\Delta_\ell(g)
=
H_{w_{[\ell+1]}} g
-
A^\ast H_{w_{[\ell]}} g.
\end{equation}
Subtracting~\eqref{eq:barx_dynamics_main} from~\eqref{eq:z_true_update_step3_short} gives
\begin{equation}\label{eq:z_minus_barx_recursive_short}
z_{\ell+1|k} - \bar x_{\ell+1|k}
=
\Delta_\ell(g),
\qquad
z_{0|k} - \bar x_{0|k} = 0.
\end{equation}
Here $H_w$ is the actual, unknown noise matrix of the fixed offline data.
Proposition~\ref{prop:window_model_realization_TAC_poly} applies to all
candidate models and gives $\Delta_\ell(g)\in(1+\gamma^\ast)\mathcal W$
for the true pair. The actual-noise identity yields a refined bound.
Since $H_{w_{[\ell]}}g,H_{w_{[\ell+1]}}g\in\mathcal W$,
the identity~\eqref{eq:Delta_true_step3_short} further gives
$\Delta_\ell(g)\in\mathcal W\oplus(-\mathcal W_A)$. Hence,
\begin{equation}\label{eq:step3_pointwise_bound_short}
z_{\ell|k}-\bar x_{\ell|k} \in\mathcal W\oplus(-\mathcal W_A)=:\mathcal E_{\mathrm{off}}, \qquad\forall\,\ell\in\mathbb N_{1:L}.
\end{equation}
By~\eqref{eq:WA_uniform_LTI},
$\mathcal E_{\mathrm{off}}\subseteq(1+\gamma^\ast)\mathcal W$.
\paragraph{Propagated mismatch}
Define the error
\begin{equation}\label{eq:e_def_TAC}
e_{\ell|k} := z_{\ell|k}-\hat x_{\ell|k}, \qquad \forall\,\ell\in\mathbb N_{0:L}.
\end{equation}
Subtracting~\eqref{eq:tildex_dynamics_main} from~\eqref{eq:barx_dynamics_main} gives
\begin{equation}\label{eq:xbar_minus_xtilde_TAC_e}
\bar x_{\ell+1|k}-\tilde x_{\ell+1|k}
=
(A^\ast+B^\ast K)\, e_{\ell|k}.
\end{equation}

The initial anchoring constraint~\eqref{eq:robust_deepc_init} allows
$z_{0|k}$ to be chosen within $\hat x_k+\mathcal E$ to accommodate the
noisy initialization. We choose $\mathcal E$ to be RPI for the error
dynamics below.

To construct $\mathcal E$, write the prediction error as
\begin{equation}\label{eq:err_decomp_TAC}
\begin{aligned}
e_{\ell|k}
&=\bigl(z_{\ell|k}-\bar x_{\ell|k}\bigr)
 + \bigl(\bar x_{\ell|k}-\tilde x_{\ell|k}\bigr)\\
&\quad+ \bigl(\tilde x_{\ell|k}-x_{\ell|k}\bigr)
 + \bigl(x_{\ell|k}-\hat x_{\ell|k}\bigr).
\end{aligned}
\end{equation}
The measurement-noise and offline mismatch bounds established above,
together with $x_{\ell|k}-\hat x_{\ell|k}=-w_{k+\ell}\in-\mathcal W$,
bound the additive terms. Define
\begin{equation}\label{eq:D_TAC_refined}
D:=\mathcal E_{\mathrm{off}}\oplus\mathcal E_w\oplus(-\mathcal W).
\end{equation}
Since each summand in~\eqref{eq:D_TAC_refined} is a compact polytope and
$\mathcal W$ is a compact polytope with $0\in\mathcal W$, the set $D$ is again
a compact polytope containing the origin. Using
$\bar x_{\ell|k}-\tilde x_{\ell|k} = (A^\ast{+}B^\ast K)\, e_{\ell-1|k}$ and the
decomposition~\eqref{eq:err_decomp_TAC}, we obtain the pointwise error
recursion
\begin{equation}\label{eq:state_map_TAC_refined}
e_{\ell|k}
= (A^\ast{+}B^\ast K)\, e_{\ell-1|k} \;+\; d_{\ell|k},
\qquad d_{\ell|k} \in D,
\end{equation}
with initialization $e_{0|k}=z_{0|k}-\hat x_k\in\mathcal E$.

Since $A^\ast+B^\ast K\in\mathcal A_K$, the corresponding robust
error set-dynamics
\begin{equation}\label{eq:e_set_dyn_refined}
e_{\ell+1|k}
\in
\mathcal A_K\, e_{\ell|k}
\ \oplus\ 
D,
\end{equation}
and the corresponding robust data-driven RPI condition
\begin{equation}\label{eq:RPI_condition_TAC_refined}
\boxed{\quad
\mathcal A_K\, \mathcal E
\ \oplus\ 
D
\ \subseteq\ 
\mathcal E.
\quad}
\end{equation}
Choose a compact polytope $\mathcal E$ satisfying~\eqref{eq:RPI_condition_TAC_refined}. Such a set can be computed using
the data-driven set-iteration procedure in~\cite{WangAngeli2026CDCsub},
with the resulting inclusion verified before implementation.
Then $e_{0|k}\in\mathcal E$ implies $e_{\ell|k}\in\mathcal E$ for
$\ell\in\mathbb N_{0:L}$ by induction. In particular, the same condition
preserves the initial constraint after a one-step horizon shift.
Consequently,
\begin{equation}\label{eq:xbar_minus_xtilde_bound_TAC}
\bar x_{\ell+1|k}-\tilde x_{\ell+1|k}
 \in (A^\ast{+}B^\ast K)\,\mathcal E,
\qquad \forall\,\ell\in\mathbb N_{0:L-1}.
\end{equation}

Since ${\mathcal A}_K$ is a polytope in the matrix space, it admits a finite
vertex set $\mathrm{vert}({\mathcal A}_K)=\{A_K^{(j)}\}_{j=1}^{N_K}$, and thus
\[
\operatorname{co}\!\bigl({\mathcal A}_K \mathcal E\bigr)
=
\operatorname{co}\!\Big(\bigcup_{j=1}^{N_K} A_K^{(j)}\mathcal E\Big),
\]
which is a polytope as the convex hull of a finite union of polytopes. We thus define the
propagated-mismatch bound
\begin{equation}\label{eq:Eon_poly_def}
\mathcal E_{\mathrm{on}}
:=
\operatorname{co}\!\bigl({\mathcal A}_K \mathcal E\bigr),
\end{equation}
and conclude from~\eqref{eq:xbar_minus_xtilde_bound_TAC} that
\begin{equation}\label{eq:step2tight}
\bar x_{\ell+1|k}-\tilde x_{\ell+1|k}
\in
\operatorname{co}\!\bigl({\mathcal A}_K \mathcal E\bigr)
=\mathcal E_{\mathrm{on}},
\quad
\forall\,\ell\in\mathbb N_{0:L-1}.
\end{equation}
This robustification replaces the unknown $(A^\ast,B^\ast)$ by the
data-consistent polytope ${\mathcal A}_K$ and yields a computable
polytopic bound.

Thus, for $\ell\in\mathbb N_{1:L}$, the component bounds are
\begin{equation}\label{eq:tube_members_TAC}
\begin{alignedat}{2}
z_{\ell|k}-\bar x_{\ell|k}&\in\mathcal E_{\mathrm{off}},
&\quad\bar x_{\ell|k}-\tilde x_{\ell|k}&\in\mathcal E_{\mathrm{on}},\\
\tilde x_{\ell|k}-x_{\ell|k}&\in\mathcal E_{w},
&\quad x_{\ell|k}-\hat x_{\ell|k}&=-w_{k+\ell}\in-\mathcal W.
\end{alignedat}
\end{equation}

\subsubsection{Tightened Constraints}\label{sec:tightened_constraints}
Combining~\eqref{eq:tube_chain_intro} with~\eqref{eq:tube_members_TAC}
gives the tightened admissible state set $\mathcal Z$ such that
\begin{equation}\label{eq:Zl_def_TAC_final}
\begin{gathered}
z_{\ell|k}\in\mathcal Z:=\mathcal X\ominus(-\mathcal E_{\mathrm{off}})
\ominus(-\mathcal E_{\mathrm{on}})\ominus(-\mathcal E_{w})\\
\Longrightarrow\quad x_{\ell|k}\in\mathcal X,
\quad\forall\,\ell\in\mathbb{N}_{1:L}.
\end{gathered}
\end{equation}

An analogous tightening applies to the input. The OCP constrains the
planned closed-loop input $
u_{\ell|k}^{\mathrm{p}}
:=
v_{\ell|k}+K z_{\ell|k}
$,
whereas the input along the fixed-plan rollout is
$
u_{\ell|k}^{\mathrm{a}}
:=
v_{\ell|k}+K \hat x_{\ell|k}
$.
Their difference is
\[
u_{\ell|k}^{\mathrm{p}}
-
u_{\ell|k}^{\mathrm{a}}
=
K\big(z_{\ell|k}-\hat x_{\ell|k}\big)
=
K e_{\ell|k}.
\]
Since $e_{\ell|k}\in\mathcal E$ by~\eqref{eq:RPI_condition_TAC_refined}, the image $K\mathcal E$ is a polytope
bounding this mismatch. We therefore introduce a tightened admissible
input set $\hat{\mathcal U}$ such that
\begin{equation}\label{eq:Ul_def_TAC_final}
\begin{gathered}
u_{\ell|k}^{\mathrm p}=v_{\ell|k}+Kz_{\ell|k}
\in\hat{\mathcal U}:=\mathcal U\ominus(-K\mathcal E)\\
\Longrightarrow\quad v_{\ell|k}+K\hat x_{\ell|k}\in\mathcal U,
\quad\forall\,\ell\in\mathbb{N}_{0:L-1}.
\end{gathered}
\end{equation}

Throughout, we assume that the tightened sets are nonempty.
\begin{assumption}\label{as:nonempty_tightened}
The tightened state and input sets satisfy
\[
\mathcal Z \neq \varnothing,
\qquad
\hat{\mathcal U} \neq \varnothing.
\]
\end{assumption}

Assumption~\ref{as:nonempty_tightened} requires the computed
tightening margins to be compatible with the original state
and input constraints. This depends on the noise-induced
tube size relative to the available constraint margins and
can be checked offline by linear-programming feasibility
tests for the computed polytopes.
\begin{remark}\label{rem:gamma-star}
Relative to the same tube construction with the true model,
the data-driven tightening introduces additional conservatism
through (i) the enclosure
$A^\ast\mathcal W\subseteq\mathcal W_A$ and
(ii) robustification over
$\mathcal I_{\mathrm{full}}^{\mathrm{poly},\ast}$
in place of the true pair $(A^\ast,B^\ast)$.
\end{remark}

\subsection{Recursive Feasibility}\label{sec:rec-feas}
We establish recursive feasibility using a simplex-preserving
coefficient shift, invariance of the initial error tube, and terminal
invariance and constraint compatibility. We first specify the feedback
certificate and terminal conditions, then give the tail-coverage
condition and verify the shifted candidate.

\subsubsection{Certified Feedback Gain}\label{sec:gain_certificate}
The gain $K$ should enforce the uniform quadratic decrease
\begin{equation}\label{eq:uniform_quadratic_decrease}
(A{+}BK) P (A{+}BK)^\top-P \preceq -\beta I,
 \forall (A,B)\in \mathcal I_{\mathrm{full}}^{\mathrm{poly},\ast},
\end{equation}
for some $P=P^\top\!\succ0$ and $\beta>0$.
This implies that the family $\mathcal A_K$ is uniformly exponentially stable.

Since $\mathcal I_{\mathrm{full}}^{\mathrm{poly},\ast}$ is a compact polytope
in the matrix space, it admits a finite vertex set
\[
\operatorname{vert}\!\bigl(\mathcal I_{\mathrm{full}}^{\mathrm{poly},\ast}\bigr)
=\{(A^{(j)},B^{(j)})\}_{j=1}^{N_v}.
\]
The block linear matrix inequality (LMI) below is affine in $(A,B)$ for fixed $(P,Y,\beta)$.
Thus, verification at the vertices extends to the whole polytope by convexity.

\begin{proposition}\label{prop:K-stab}
If there exist $P=P^\top\succ0$, $Y\in\mathbb{R}^{m\times n}$, and $\beta>0$ such that
\begin{equation}\label{eq:poly_decay_SDP}
\begin{aligned}
&\begin{bmatrix}
P-\beta I & S^{(j)}\\
\bigl(S^{(j)}\bigr)^\top & P
\end{bmatrix}\succeq 0,\qquad \forall\, j \in \mathbb{N}_{1:N_v},\\[-.25em]
& S^{(j)} := A^{(j)}P + B^{(j)}Y,
\end{aligned}
\end{equation}
then with $K:=Y P^{-1}$ we have
\[
(A{+}BK)\,P\,(A{+}BK)^\top - P + \beta I \preceq 0,
\ \forall\ (A,B)\in\mathcal I_{\mathrm{full}}^{\mathrm{poly},\ast}.
\]
In particular, the true closed-loop matrix $A^\ast{+}B^\ast K$ is Schur.
\end{proposition}

\begin{proof}
The claim follows from Proposition~2 in~\cite{WangAngeli2026CDCsub}. Since $(A^\ast,B^\ast)\in\mathcal I_{\mathrm{full}}^{\mathrm{poly},\ast}$, the inequality~\eqref{eq:uniform_quadratic_decrease} holds at the true pair, i.e.,
\[
(A^\ast{+}B^\ast K) P (A^\ast{+}B^\ast K)^\top-P \preceq -\beta I \prec 0,
\]
which implies that $A^\ast{+}B^\ast K$ is Schur.
\end{proof}

For a prescribed gain $K$, impose $Y=KP$ in~\eqref{eq:poly_decay_SDP}.
The resulting semidefinite program (SDP) verifies the required certificate on
$\mathcal I_{\mathrm{full}}^{\mathrm{poly},\ast}$ without changing $K$.

\subsubsection{Terminal Invariance and Constraint Compatibility}\label{sec:terminal_design}
We define the terminal set as a data-driven RPI polytope under a local
affine feedback law
\begin{equation}\label{eq:terminal_affine_law}
u \;=\; K z + v,\qquad v\in\mathcal V,
\end{equation}
where $\mathcal V\subset\mathbb R^{m}$ is a prescribed compact convex polytope.
Recall that each Hankel-predicted step $z_{\ell|k}\mapsto z_{\ell+1|k}$ satisfies
the update~\eqref{eq:z_true_update_step3_short} for the true model, with the planned input
$u^{\mathrm p}_{\ell|k}$ and an additive residual $\Delta_\ell(g)\in\mathcal E_{\mathrm{off}}$.

In a standard shift argument one would append the pure local law $u=Kz$ at the
terminal step. Under the simplex restriction $g\in\Delta^{T-L}$, however, this
append is not free: the deviation induced by the last block row of the Hankel
parametrization is
\[
v(g):=\bigl(H_{u_{[L]}}-K\hat H_{x_{[L]}}\bigr)g
\in
\operatorname{co}\!\{
\bigl(H_{u_{[L]}}-K\hat H_{x_{[L]}}\bigr)e_i
\}_{i=1}^{T-L}.
\]
Here, $e_i$ is the $i$th canonical basis vector.
Thus, $v(g)=0$ (and hence the pure law $u=Kz$) cannot be guaranteed for all
feasible $g$. We therefore
admit a bounded terminal deviation $v\in\mathcal V$ in~\eqref{eq:terminal_affine_law}
and require that all simplex-admissible terminal deviations generated by the
last block row are contained in $\mathcal V$.

\begin{assumption}\label{as:TCC_tail_V}
Fix a compact convex polytope $\mathcal V\subset\mathbb R^{m}$.
The offline data satisfy
\begin{equation}\label{eq:V_tail_simplex}
\bigl(H_{u_{[L]}}-K\hat H_{x_{[L]}}\bigr)g \ \in\ \mathcal V,
\qquad \forall\, g\in\Delta^{T-L}.
\end{equation}
Equivalently,
\begin{equation}\label{eq:V_tail_vertices}
\bigl(H_{u_{[L]}}-K\hat H_{x_{[L]}}\bigr)e_i \ \in\ \mathcal V,
\qquad i=1,\ldots,T-L.
\end{equation}
\end{assumption}

Assumption~\ref{as:TCC_tail_V} bounds the deviation appended after
shifting the prediction. The vertex conditions~\eqref{eq:V_tail_vertices}
can be checked directly for given offline data. One sufficient way to satisfy
them is to collect an offline input--state trajectory under
$u_i^d=K\hat x_i^d+v_i^d$, with $v_i^d\in\mathcal V$.
Indeed, for every $g\in\Delta^{T-L}$,
\begin{equation}\label{eq:collected_terminal_deviation}
\bigl(H_{u_{[L]}}-K\hat H_{x_{[L]}}\bigr)g
=\sum_{j=0}^{T-L-1}g_{j+1}v_{j+L}^d
\in\mathcal V.
\end{equation}
The last block row involves the indices $L,\ldots,T-1$.
Appendix~\ref{app:TCC_tail_verify} gives one sufficient data-acquisition procedure.

In particular, at the terminal step the predicted closed-loop dynamics under~\eqref{eq:terminal_affine_law} can be written as
\begin{equation}\label{eq:terminal_dyn_affine}
z^{+} \in \mathcal A_K z \ \oplus\ \Omega,
\qquad
\Omega := \mathcal E_{\mathrm{off}} \ \oplus\ \mathcal B\,\mathcal V,
\end{equation}
where $\mathcal B:=\{B\mid (A,B)\in\mathcal I_{\mathrm{full}}^{\mathrm{poly},\ast}\}$.
Since $\mathcal I_{\mathrm{full}}^{\mathrm{poly},\ast}$ is a compact polytope in the matrix space,
$\mathcal B$ is compact and
\[
\mathcal B\,\mathcal V
:= \operatorname{co}\{Bv\mid B\in\mathcal B,\ v\in\mathcal V\}
= \operatorname{co}\Bigl(\bigcup_{j=1}^{N_v} B^{(j)}\mathcal V\Bigr)
\]
is a compact convex polytope.

Let $\mathcal M_\infty$ denote the minimal RPI set for~\eqref{eq:terminal_dyn_affine}, i.e.,
\[
\mathcal M_\infty
:= \bigoplus_{i=0}^{\infty}
\operatorname{co}\bigl(\mathcal A_K^{\,i}\Omega\bigr).
\]
Here, $\mathcal A_K^0\Omega:=\Omega$ and
$\mathcal A_K^{i+1}\Omega:=\mathcal A_K(\mathcal A_K^i\Omega)$.
The certified gain makes $\mathcal A_K$ uniformly exponentially
stable. Together with compactness of $\Omega$, this ensures that
$\mathcal M_\infty$ is well-defined, convex, and compact.

To complete the terminal-set construction, we require the existence of a
compact convex polytope that contains the minimal RPI set and is compatible
with the tightened state and input constraints.

\begin{assumption}\label{as:terminal_set}
There exists a compact convex polytope $\mathcal Z_f$ such that
\begin{subequations}\label{eq:terminal_set_final}
\begin{align}
\mathcal M_\infty \subseteq \mathcal Z_f &\subseteq \mathcal Z,
\label{eq:terminal_Minf_subset}
\\
K\mathcal Z_f \oplus \mathcal V &\subseteq \hat{\mathcal U},
\label{eq:terminal_input_tight}
\\
\mathcal A_K \mathcal Z_f \oplus \Omega
&\subseteq \mathcal Z_f.
\label{eq:terminal_RPI_condition}
\end{align}
\end{subequations}
\end{assumption}

Condition~\eqref{eq:terminal_RPI_condition} states that $\mathcal Z_f$ is RPI
for the closed-loop family $\mathcal A_K$ under the compact set
$\Omega$, while~\eqref{eq:terminal_Minf_subset}--\eqref{eq:terminal_input_tight}
enforce compatibility with the tightened state and input constraints at the
terminal step.

Beyond the nonemptiness requirement in
Assumption~\ref{as:nonempty_tightened}, terminal compatibility
also depends on the terminal-deviation set $\mathcal V$.
For fixed $K$ and certified model sets, enlarging $\mathcal V$
enlarges the disturbance set $\Omega$
in~\eqref{eq:terminal_dyn_affine} and leaves less input margin
in~\eqref{eq:terminal_input_tight}.
Thus, the choice of $\mathcal V$ balances the permitted
excitation during data collection against terminal
constraint compatibility.

By construction, if the terminal nominal state $z_{L|k}$ lies in $\mathcal Z_f$,
then under the local law $u = K z + v$ with any $v\in\mathcal V$ every successor
remains in $\mathcal Z_f$ and the tightened input constraint
$Kz+v\in\hat{\mathcal U}$ is satisfied for all future times and all
$(A,B)\in\mathcal I_{\mathrm{full}}^{\mathrm{poly},\ast}$ and additive residuals
in $\mathcal E_{\mathrm{off}}$. Therefore, choosing $\mathcal Z_f$
as terminal set provides a robust data-driven terminal RPI.

The terminal RPI polytope can be constructed using the data-driven
set-iteration procedure in~\cite{WangAngeli2026CDCsub}. Its limit is
$\mathcal M_\infty$; a selected polytope $\mathcal Z_f$ must satisfy
the invariance and compatibility conditions~\eqref{eq:terminal_set_final}.
Proposition~\ref{prop:K-stab} provides the uniform stability certificate
used in the RPI-set construction. The offline design selects $\mathcal E$
and $\mathcal Z_f$ to satisfy~\eqref{eq:RPI_condition_TAC_refined}
and~\eqref{eq:terminal_set_final}, respectively. The latter conditions
also enforce compatibility with the tightened state and input constraints.

\subsubsection{Tail Coverage for the Coefficient Shift}\label{sec:tail_shift}
In addition to the input PE condition in
Assumption~\ref{as:available-trajectory}, the coefficient shift requires
a convex representation of the final data window. We impose the following
tail convex-coverage condition to keep the shifted coefficient in the unit simplex.
\begin{assumption}\label{as:TCC_tail}
Let $L,T\in\mathbb{Z}_{>0}$ with $T>L$.
There exists a coefficient $h\in\Delta^{T-L}$ such that
\begin{equation}\label{eq:TCC_tail}
\begin{bmatrix}
H_{u_{[0:L-1]}}\\[.3mm]
\hat H_{x_{[0:L-1]}}
\end{bmatrix} h
=
\begin{bmatrix}
u^d_{[T-L,T-1]}\\[.3mm]
\hat x^d_{[T-L,T-1]}
\end{bmatrix}.
\end{equation}
Equivalently, the tail window
$\big(u^d_{[T-L,T-1]},\hat x^d_{[T-L,T-1]}\big)$
lies in the convex hull of the $T-L$ prefix columns from the
depth-$(L+1)$ Hankel matrices.
\end{assumption}

\begin{remark}\label{rem:TCC_tail_role}
Assumption~\ref{as:TCC_tail} is an exact, offline-checkable tail
convex-coverage condition. PE ensures linear spanning in the clean-data
setting, but does not imply this convex representation. The tail condition
is used to close the
coefficient shift within the simplex. Appendix~\ref{app:TCC_tail_verify}
gives the linear program (LP) test and sufficient stochastic conditions for finite
termination of Algorithm~\ref{alg:TCC}.
\end{remark}

\subsubsection{Recursive Feasibility Guarantee}\label{sec:recursive_guarantee}
Assumption~\ref{as:TCC_tail} provides the simplex-preserving Hankel
shift, while Assumption~\ref{as:TCC_tail_V} bounds the appended input
deviation. The error-tube condition~\eqref{eq:RPI_condition_TAC_refined}
preserves the initial constraint, and Assumption~\ref{as:terminal_set}
ensures terminal invariance and admissibility. These ingredients yield
recursive feasibility.

\begin{theorem}\label{thm:rec_feas_shift}
Suppose Assumptions~\ref{as:available-trajectory}--\ref{as:TCC_tail} hold,
and let $\mathcal E$ satisfy~\eqref{eq:RPI_condition_TAC_refined}.
If the measured state at time $k$ satisfies
$\hat x_k\in\hat{\mathcal X}_L^{\mathrm{feas}}$ (equivalently, the
OCP~\eqref{eq:robust_deepc} with parameter $\hat x_k$ is feasible),
then the OCP~\eqref{eq:robust_deepc} is also feasible at time $k{+}1$,
i.e., $\hat x_{k+1}\in\hat{\mathcal X}_L^{\mathrm{feas}}$.
In particular, if $\hat x_0\in\hat{\mathcal X}_L^{\mathrm{feas}}$ then
$\hat x_k\in\hat{\mathcal X}_L^{\mathrm{feas}}$ for all $k\in\mathbb{N}_0$.
\end{theorem}

\begin{proof}
We construct a shifted candidate and verify its initial, terminal,
and stage constraints in that order.

\textbf{(i) Shifted candidate and Hankel parametrization.}
Since $\hat x_k\in\hat{\mathcal X}_L^{\mathrm{feas}}$ at time $k$, the OCP admits
an optimal solution $(Z_k^\star,V_k^\star,g^\star)$, with $g^\star\in\Delta^{T-L}$.
The associated planned inputs are
$u_{\ell|k}^{\mathrm p,\star}=Kz_{\ell|k}^\star+v_{\ell|k}^\star$. By feasibility at time $k$, all tightened constraints are satisfied and
$z_{L|k}^\star\in\mathcal Z_f$.

Let
$J\in\mathbb{R}^{(T-L)\times(T-L)}$ be the forward shift matrix defined by
$Je_i=e_{i+1}$ for $i=1,\ldots,T{-}L{-}1$ and $Je_{T-L}=0$.
Under Assumption~\ref{as:TCC_tail}, pick $h\in\Delta^{T-L}$ satisfying~\eqref{eq:TCC_tail} and define
\begin{equation}\label{eq:gplus_def_tail}
g^+ := Jg^\star + g_{T-L}^\star h.
\end{equation}
The construction applies to every feasible coefficient; uniqueness of the
shifted candidate is not required.
Then $g^+\ge 0$ and $\mathbf 1^\top g^+ = (\mathbf 1^\top g^\star-g_{T-L}^\star)+g_{T-L}^\star=1$, hence $g^+\in\Delta^{T-L}$.
Moreover, a direct calculation using~\eqref{eq:TCC_tail} yields, for all $\ell\in\mathbb N_{0:L-1}$,
\begin{equation}\label{eq:Hankel_shift_hold_tail}
H_{u_{[\ell]}}g^+  = H_{u_{[\ell+1]}}g^\star,\qquad
\hat H_{x_{[\ell]}}g^+  = \hat H_{x_{[\ell+1]}}g^\star.
\end{equation}

Define the full candidate, including its terminal state and auxiliary
terminal input deviation, by
\begin{equation}\label{eq:behav_real_proof}
\begin{bmatrix}
v_{[0|k+1, L|k+1]}^c\\[.3mm]
z_{[0|k+1, L|k+1]}^c
\end{bmatrix}
=
\begin{bmatrix}
H_u-\tilde K \hat H_x\\[.3mm]
\hat H_x
\end{bmatrix} g^+,\quad g^+\in\Delta^{T-L}.
\end{equation}
Thus the Hankel equalities and simplex constraint hold by construction.
The identities~\eqref{eq:Hankel_shift_hold_tail} give
\begin{equation}\label{eq:shift_state_TAC}
z_{\ell|k+1}^c := z_{\ell+1|k}^\star,\qquad \forall\,\ell \in \mathbb{N}_{0:L-1},
\end{equation}
and the candidate input sequence satisfies
\begin{equation}\label{eq:shift_input_TAC_affine}
\begin{aligned}
&u_{[0|k+1,\,L-1|k+1]}^{\mathrm p,c}\\
&\quad:=
\operatorname{col}\bigl(
u_{1|k}^{\mathrm p,\star},\ldots,u_{L-1|k}^{\mathrm p,\star},\;
K z_{L|k}^\star + v_{L-1|k+1}^c
\bigr).
\end{aligned}
\end{equation}
In particular, the appended terminal deviation is
\begin{equation}\label{eq:v_append_candidate}
v_{L-1|k+1}^c
:=\bigl(H_{u_{[L-1]}}-K\hat H_{x_{[L-1]}}\bigr)g^+.
\end{equation}
Then, by~\eqref{eq:Hankel_shift_hold_tail} with $\ell=L{-}1$,
\[
v_{L-1|k+1}^c
=\bigl(H_{u_{[L]}}-K\hat H_{x_{[L]}}\bigr)g^\star \in \mathcal V,
\]
where the inclusion follows from Assumption~\ref{as:TCC_tail_V} and $g^\star\in\Delta^{T-L}$.

\textbf{(ii) Initial constraint.}
At time $k$, feasibility enforces $e_{0|k} = z_{0|k}^\star - \hat x_k \in \mathcal E$.
By~\eqref{eq:shift_state_TAC}, $z_{0|k+1}^c = z_{1|k}^\star$. Since the applied input is $v_{0|k}^\star+K\hat x_k$,
we have $\hat x_{k+1}=\hat x_{1|k}$, and hence
\[
e_{0|k+1} := z_{0|k+1}^c - \hat x_{k+1}
= z_{1|k}^\star - \hat x_{1|k}
=: e_{1|k}.
\]
By the error set dynamics~\eqref{eq:e_set_dyn_refined} and the RPI condition~\eqref{eq:RPI_condition_TAC_refined}, $e_{1|k}\in\mathcal E$ whenever $e_{0|k}\in\mathcal E$.
Therefore $z_{0|k+1}^c - \hat x_{k+1} \in \mathcal E$, and the initial constraint is satisfied at time $k{+}1$.

\smallskip
\textbf{(iii) Terminal constraint.}
Since $g^+\in\Delta^{T-L}$, the identity~\eqref{eq:Delta_true_step3_short}
for the actual offline noise gives
\begin{equation}\label{eq:terminal_residual_direct}
\begin{aligned}
R_{L-1}(A^\ast,B^\ast)g^+
&=H_{w_{[L]}}g^+-A^\ast H_{w_{[L-1]}}g^+\\
&\in\mathcal W\oplus(-\mathcal W_A)=\mathcal E_{\mathrm{off}}.
\end{aligned}
\end{equation}
The shift identities and Assumption~\ref{as:TCC_tail_V} yield
$z_{L-1|k+1}^c=z_{L|k}^\star$ and
$v_{L-1|k+1}^c\in\mathcal V$. Therefore,
\begin{equation}\label{eq:terminal_successor_direct}
\begin{aligned}
z_{L|k+1}^c
&=(A^\ast+B^\ast K)z_{L|k}^\star+B^\ast v_{L-1|k+1}^c\\
&\quad+R_{L-1}(A^\ast,B^\ast)g^+\\
&\in\mathcal A_K z_{L|k}^\star\oplus\Omega
\subseteq\mathcal Z_f,
\end{aligned}
\end{equation}
where the last inclusion uses Assumption~\ref{as:terminal_set}.

\smallskip
\textbf{(iv) Stage constraints.}
For $\ell=0,\ldots,L{-}2$, the shift yields
$z_{\ell|k+1}^c = z_{\ell+1|k}^\star \in \mathcal Z$ and
$u_{\ell|k+1}^{\mathrm p,c} = u_{\ell+1|k}^{\mathrm p,\star} \in \hat{\mathcal U}$.
For $\ell = L{-}1$, we have
$z_{L-1|k+1}^c = z_{L|k}^\star \in \mathcal Z_f \subseteq \mathcal Z$ and
\[
u_{L-1|k+1}^{\mathrm p,c} = K z_{L|k}^\star + v_{L-1|k+1}^c,
\qquad v_{L-1|k+1}^c\in\mathcal V.
\]
By~\eqref{eq:terminal_set_final}, $K\mathcal Z_f \oplus \mathcal V \subseteq \hat{\mathcal U}$;
hence $u_{L-1|k+1}^{\mathrm p,c}\in\hat{\mathcal U}$.
Moreover, $z_{L|k+1}^c\in\mathcal Z$ follows from $\mathcal Z_f\subseteq\mathcal Z$. Thus all tightened stage constraints are satisfied at time $k{+}1$.

\smallskip
Items (i)--(iv) show that $(Z_{k+1}^c, V_{k+1}^c, g^+)$
satisfies all constraints of the OCP at time $k{+}1$.
Hence the OCP is feasible at $k{+}1$.
\end{proof}

\subsection{Closed-loop Constraint Satisfaction}\label{Constraint Satisfaction}
The following result shows that the implemented input and the actual closed-loop state respect the original constraints at all times.
\begin{theorem}\label{thm:cl_consat_TAC}
Under the hypotheses of Theorem~\ref{thm:rec_feas_shift},
if $x_0\in\mathcal X_L^{\mathrm{feas}}$, then, under the control law $
u_k \;=\; v_{0|k}^\star \;+\; K\,\hat x_k
$,
the closed-loop input and state satisfy, for all $k\in\mathbb{N}_0$,
\[
u_k \in \mathcal U, \qquad x_{k+1} \in \mathcal X.
\]
\end{theorem}

\begin{proof}
Assume $x_0\in\mathcal X_L^{\mathrm{feas}}$. By definition~\eqref{eq:X_L_feas_def}, this implies $\hat x_0=x_0+w_0\in\hat{\mathcal X}_L^{\mathrm{feas}}$
for all $w_0\in\mathcal W$, and in particular for the realized $w_0$.
Hence the OCP~\eqref{eq:robust_deepc} is feasible at time $k=0$.
By Theorem~\ref{thm:rec_feas_shift}, it remains feasible for all $k\in\mathbb N_0$.

Fix any $k\in\mathbb N_0$ and let $(Z_k^\star,V_k^\star,g^\star)$ be an optimal
solution at time $k$. By the input-side tightening in~\eqref{eq:Ul_def_TAC_final} with $\ell=0$ and
$\hat x_{0|k}=\hat x_k$, feasibility implies
\[
v_{0|k}^\star + K z_{0|k}^\star \in \hat{\mathcal U}
\ \Longrightarrow\ 
v_{0|k}^\star + K \hat x_k \in \mathcal U,
\]
so that $u_k=v_{0|k}^\star+K\hat x_k\in\mathcal U$. By the state-side tightening in~\eqref{eq:Zl_def_TAC_final}, feasibility implies
$z_{1|k}^\star\in\mathcal Z$, and therefore
\[
x_{k+1}=x_{1|k}\in\mathcal X.
\]
Since this holds for all $k\in\mathbb N_0$, the closed-loop satisfies
$u_k\in\mathcal U$ and $x_{k+1}\in\mathcal X$ for all $k\in\mathbb N_0$.
\end{proof}

\subsection{Practical Input-to-State Stability}\label{ISS}
In this section, we establish practical ISS of the true closed-loop system on
$\mathcal X_L^{\mathrm{feas}}$. We use the optimal predicted cost as the
Lyapunov candidate. First, we derive a one-step cost estimate and use it to
bound the optimal finite-horizon value. Second, we apply the same estimate
to the shifted feasible trajectory to obtain a one-step decrease inequality
for $V(\hat x_k)$. Finally, using the anchoring constraint
$z_{0|k}^\star-\hat x_k\in\mathcal E$ and $\hat x_k=x_k+w_k$, we pass from
the value estimate to the true state. All constants are computable and
summarized in Appendix~\ref{app:constants}.

For the fixed offline data and feedback gain $K$, define
\begin{equation}\label{eq:stability_data_amplitudes}
\begin{aligned}
d_i^d&:=u_i^d-K\hat x_i^d,\qquad
\epsilon:=\max_{w\in\mathcal W}\|w\|_2,\\
\nu_{\mathrm d}&:=\max_{0\le i<T}\|d_i^d\|_2,\qquad
\nu_{\mathrm f}:=\max_{L\le i<T}\|d_i^d\|_2.
\end{aligned}
\end{equation}
The simplex parametrization and the shift identities give
\begin{equation}\label{eq:stability_deviation_bounds}
\begin{aligned}
\|v_{\ell|k}\|_2&\le\nu_{\mathrm d},
&&\ell\in\mathbb N_{0:L},\\
\|v_{L-1|k+1}^c\|_2&\le\nu_{\mathrm f}
\le\max_{v\in\mathcal V}\|v\|_2,
\end{aligned}
\end{equation}
where the last inequality follows from
Assumption~\ref{as:TCC_tail_V}.

We select the terminal weight $P_L$ consistently with the stabilizing gain $K$.
Let $(K,P,\beta)$ satisfy the certificate in
Proposition~\ref{prop:K-stab}. For every certified model, write
$F:=A+BK$ and set $S:=Q+K^\top RK$. The certificate gives
\begin{equation}\label{eq:dual_lyap_final}
FPF^\top\preceq P-\beta I,\qquad P\succ0,\quad\beta>0.
\end{equation}
Select the terminal weight as
\begin{equation}\label{eq:terminal_weight_choice}
\begin{aligned}
\lambda&:=\frac{\beta}{\lambda_{\max}(P)},\qquad P_L:=c_pP^{-1},\\
c_p&>\lambda^{-1}
\lambda_{\max}\!\left(P^{1/2}SP^{1/2}\right).
\end{aligned}
\end{equation}
Then $0<\lambda\le1$ and
$\|P^{-1/2}FP^{1/2}\|_2^2\le1-\lambda$, so
$F^\top P_LF\preceq(1-\lambda)P_L$. In particular, with
$d:=\lambda_{\min}(\lambda P_L-S)>0$,
\begin{equation}\label{eq:primal_lyap_final}
P_L-F^\top P_LF-S\succeq dI.
\end{equation}

We next use~\eqref{eq:primal_lyap_final} to derive a one-step cost estimate
for the predicted dynamics.

\begin{lemma}\label{lem:one_step_cost}
There exist computable constants $c_w,c_v>0$, specified in
Appendix~\ref{app:one_step_constants}, such that, with
\begin{equation}\label{eq:cost_bias_function}
c(r):=c_w\epsilon^2+c_vr^2,\qquad r\ge0,
\end{equation}
every certified model $(A,B)$ satisfies
\begingroup
\thinmuskip=2mu\medmuskip=1mu\thickmuskip=1mu\relax
\begin{IEEEeqnarray}{l}\label{eq:ISStype_decrease}
J_{\mathrm f}((A+BK)z+Bv+\rho)-J_{\mathrm f}(z)+J_{\mathrm s}(z,Kz+v)\le c(\|v\|_2)\IEEEeqnarraynumspace
\end{IEEEeqnarray}
\endgroup
for all $z\in\mathbb R^n$, $v\in\mathbb R^m$, and
$\rho\in(1+\gamma^\ast)\mathcal W$.
\end{lemma}
\begin{proof}
Expanding the quadratic costs, applying~\eqref{eq:primal_lyap_final},
and bounding the cross term by Young's inequality gives the estimate.
The expansion and uniform constants are given in
Appendix~\ref{app:one_step_constants}.
\end{proof}

We now use this estimate to bound the finite-horizon cost.
Let $F^\ast:=A^\ast+B^\ast K$. By
Proposition~\ref{prop:window_model_realization_TAC_poly}, every
feasible prediction satisfies, for $\ell\in\mathbb N_{0:L-1}$,
\begin{equation}\label{eq:prediction_cost_dynamics}
\begin{aligned}
z_{\ell+1|k}&=F^\ast z_{\ell|k}+B^\ast v_{\ell|k}+\rho_{\ell|k},\\
\rho_{\ell|k}&:=R_\ell(A^\ast,B^\ast)g
\in(1+\gamma^\ast)\mathcal W.
\end{aligned}
\end{equation}
Summing~\eqref{eq:ISStype_decrease} over $\ell=0,\ldots,L-1$
and using~\eqref{eq:stability_deviation_bounds} yields
\begingroup
\thinmuskip=2mu\medmuskip=2mu\thickmuskip=2mu\relax
\begin{IEEEeqnarray}{l}\label{eq:horizon_cost_bound}
\sum_{\ell=0}^{L-1}J_{\mathrm s}(z_{\ell|k},Kz_{\ell|k}+v_{\ell|k})+J_{\mathrm f}(z_{L|k})\le J_{\mathrm f}(z_{0|k})+Lc(\nu_{\mathrm d}).\IEEEeqnarraynumspace
\end{IEEEeqnarray}
\endgroup
Using the quadratic stage and terminal costs in~\eqref{eq:QR_costs}, define the
optimal finite-horizon value at time $k$ as
\begin{equation}\label{eq:value_function}
V(\hat x_k):={}\sum_{\ell=0}^{L-1} J_{\mathrm s}\bigl(z_{\ell|k}^\star,Kz_{\ell|k}^\star+v_{\ell|k}^\star\bigr) +J_{\mathrm f}(z_{L|k}^\star).
\end{equation}
Since $L\ge1$, the nonnegative stage and terminal costs
and~\eqref{eq:horizon_cost_bound} imply
\begin{equation}\label{eq:value_comparison}
(z_{0|k}^\star)^\top Qz_{0|k}^\star
\le V(\hat x_k)
\le(z_{0|k}^\star)^\top P_Lz_{0|k}^\star+b,
\end{equation}
where $b:=Lc(\nu_{\mathrm d})$.

Combining the value comparison above with the shifted feasible trajectory,
we obtain the following practical ISS result. For the state estimate, let
\begin{equation}\label{eq:stability_constants}
\begin{aligned}
q&:=\lambda_{\min}(Q),\quad p:=\lambda_{\max}(P_L),
\quad M:=\sqrt{p/q},\\
r_{\mathcal E}&:=\max_{e\in\mathcal E}\|e\|_2,
\quad c_{\mathrm f}:=c(\nu_{\mathrm f}),\\
\kappa&:=\lambda_{\min}\!\left(P_L^{-1/2}QP_L^{-1/2}\right),
\quad\sigma:=1-\kappa.
\end{aligned}
\end{equation}
The strict margin in~\eqref{eq:primal_lyap_final} implies
$Q\prec P_L$, hence $0<\kappa<1$ and $0<\sigma<1$.

\begin{theorem}\label{thm:ISS_value}
Let $L\ge1$ and the hypotheses of
Theorem~\ref{thm:rec_feas_shift} hold. Suppose the feedback
certificate in Proposition~\ref{prop:K-stab} is feasible and select
$P_L$ by~\eqref{eq:terminal_weight_choice}. For every
$x_0\in\mathcal X_L^{\mathrm{feas}}$, the closed loop under
Algorithm~\ref{alg:TRDDPC} satisfies the practical ISS estimate
\begin{equation}\label{eq:ISS_true_state_bound_value}
\begin{aligned}
\|x_k\|_2\le{}&M\sigma^{k/2}\|x_0\|_2+(1+M)\eta_k\\
&+(1+M)r_{\mathcal E}
+\sqrt{\frac{b+c_{\mathrm f}/\kappa}{q}},
\end{aligned}
\end{equation}
where $\eta_k:=\max_{0\le i\le k}\|w_i\|_2$ and
$k\in\mathbb N_0$. In particular,
\begin{equation}\label{eq:ISS_ultimate_bound}
\limsup_{k\to\infty}\|x_k\|_2
\le r_{\mathcal E}+\epsilon
+\sqrt{\frac{b+c_{\mathrm f}/\kappa}{q}}.
\end{equation}
\end{theorem}

The estimate separates a decaying transient, the effect of online
measurement noise, and a design-dependent residual determined by the
fixed offline data and error tube.

\begin{proof}
By Theorem~\ref{thm:rec_feas_shift}, the OCP remains feasible.
We next derive the shift inequality for $V(\hat x_k)$.
The shifted feasible candidate satisfies
\begin{equation}\label{eq:terminal_shift_final}
\begin{aligned}
z_{L|k+1}^c
&=F^\ast z_{L|k}^\star+B^\ast v_{L-1|k+1}^c+\rho_{L|k}^c,\\
\rho_{L|k}^c
&:=R_{L-1}(A^\ast,B^\ast)g^+
\in(1+\gamma^\ast)\mathcal W.
\end{aligned}
\end{equation}
The $L-1$ common stage costs cancel in the shifted-cost comparison,
leaving the appended stage cost
$J_{\mathrm s}(z_{L|k}^\star,Kz_{L|k}^\star+v_{L-1|k+1}^c)$
and the terminal-cost increment. Online optimality,
Lemma~\ref{lem:one_step_cost}, and
$\|v_{L-1|k+1}^c\|_2\le\nu_{\mathrm f}$ give
\begin{equation}\label{eq:Vz_shift_ineq_final}
\begin{aligned}
V(\hat x_{k+1})-V(\hat x_k)
\le{}&-J_{\mathrm s}(z_{0|k}^\star,u_{0|k}^{\mathrm p,\star})
+c_{\mathrm f},
\end{aligned}
\end{equation}
where $u_{0|k}^{\mathrm p,\star}:=Kz_{0|k}^\star+v_{0|k}^\star$.
By~\eqref{eq:value_comparison} and $Q\succeq\kappa P_L$,
\[
\begin{aligned}
J_{\mathrm s}(z_{0|k}^\star,u_{0|k}^{\mathrm p,\star})
&\ge(z_{0|k}^\star)^\top Qz_{0|k}^\star \ge\kappa(z_{0|k}^\star)^\top P_Lz_{0|k}^\star\\
&\ge\kappa\bigl(V(\hat x_k)-b\bigr).
\end{aligned}
\]
Consequently,
\begin{equation}\label{eq:value_recursion}
V(\hat x_{k+1})\le\sigma V(\hat x_k)+\kappa b+c_{\mathrm f}.
\end{equation}
Iteration yields
\begin{equation}\label{eq:value_iteration}
V(\hat x_k)\le\sigma^kV(\hat x_0)
+(1-\sigma^k)\left(b+\frac{c_{\mathrm f}}{\kappa}\right).
\end{equation}

We next relate $z_{0|k}^\star$ to the true state $x_k$.
The initial tube constraint gives
$z_{0|k}^\star=x_k+w_k+e_{0|k}$ with $e_{0|k}\in\mathcal E$.
The lower bound in~\eqref{eq:value_comparison} therefore implies
\begin{equation}\label{eq:value_to_true_state}
\|x_k\|_2\le\sqrt{\frac{V(\hat x_k)}q}
+r_{\mathcal E}+\|w_k\|_2.
\end{equation}
At $k=0$, the upper bound gives
\[
V(\hat x_0)\le
p\bigl(\|x_0\|_2+r_{\mathcal E}+\|w_0\|_2\bigr)^2+b.
\]
Substituting into~\eqref{eq:value_iteration} yields
\begin{equation}\label{eq:value_initial_estimate}
V(\hat x_k)\le{}p\sigma^k \bigl(\|x_0\|_2+r_{\mathcal E}+\|w_0\|_2\bigr)^2 +b+\frac{1-\sigma^k}{\kappa}c_{\mathrm f}.
\end{equation}
Combining~\eqref{eq:value_to_true_state} and~\eqref{eq:value_initial_estimate},
using $\sqrt{a+b}\le\sqrt a+\sqrt b$ for $a,b\ge0$, and
$\|w_0\|_2,\|w_k\|_2\le\eta_k$
gives~\eqref{eq:ISS_true_state_bound_value}. Finally, taking upper limits
directly in~\eqref{eq:value_iteration} and~\eqref{eq:value_to_true_state}
gives~\eqref{eq:ISS_ultimate_bound}.
\end{proof}

\begin{remark}\label{rem:acquisition_amplitude}
The residual neighborhood can also be bounded in terms
of the offline input-deviation amplitude.
Under the hypotheses of Theorem~\ref{thm:ISS_value},
the implemented dynamics and
$\|v_{0|k}^\star\|_2\le\nu_{\mathrm d}$ yield
\begin{equation}\label{eq:acquisition_radius}
\limsup_{k\to\infty}\|x_k\|_2
\le C_v\nu_{\mathrm d}+C_w\epsilon,
\end{equation}
with computable constants given in
Appendix~\ref{app:constants}.
If the offline data satisfy
$\|u_i^d-K\hat x_i^d\|_2\le\nu$ for all
$i\in\mathbb N_{0:T-1}$, then $\nu_{\mathrm d}\le\nu$.
Consequently, the ultimate bound
$C_v\nu+C_w\epsilon$ tends to zero as
$\nu,\epsilon\to0$ along initially feasible designs
with a common feedback certificate and uniform
certified model bounds.
\end{remark}

\section{Input--Output Implementation}\label{sec:io_implementation}
Consider a controllable and observable single-input single-output LTI system of known order
$n$, with output measurements $\hat y_k=y_k+\nu_k$, where
$\nu_k\in\mathcal W_y$. Define the history vectors
\begin{equation}\label{eq:io_history}
\begin{aligned}
\xi_k&:=\operatorname{col}\bigl(u_{[k-n,k-1]},y_{[k-n,k-1]}\bigr),\\
\hat\xi_k&:=\operatorname{col}\bigl(u_{[k-n,k-1]},\hat y_{[k-n,k-1]}\bigr).
\end{aligned}
\end{equation}
With this history length, $\xi_k\in\mathbb R^{2n}$ admits a controllable
LTI realization
\begin{equation}\label{eq:io_history_dynamics}
\xi_{k+1}=A_\xi\xi_k+B_\xi u_k;
\end{equation}
see~\cite[Lemma~13]{Berberich2021TerminalIO}. The realization matrices are unknown. The measured history satisfies
\begin{equation}\label{eq:io_history_noise}
\begin{aligned}
\hat\xi_k&=\xi_k+w_k^\xi,\\
w_k^\xi&=\operatorname{col}(0_n,\nu_{k-n},\ldots,\nu_{k-1})
\in\{0_n\}\times\mathcal W_y^n.
\end{aligned}
\end{equation}
Choose a known compact convex polytope $\mathcal W_\xi$ such that
\begin{equation}\label{eq:io_noise_enclosure}
\{0_n\}\times\mathcal W_y^n\subseteq\mathcal W_\xi,
\qquad 0\in\operatorname{int}\mathcal W_\xi.
\end{equation}
Alternatively, the exact history-noise set may be used:
\begin{equation}\label{eq:io_noise_channel}
\mathcal W_\xi=G_\eta\mathcal N,\qquad
G_\eta=\begin{bmatrix}0_{n\times n}\\I_n\end{bmatrix},
\qquad \mathcal N=\mathcal W_y^n,
\end{equation}
where $\mathcal N\subset\mathbb R^n$ is a compact convex polytope
with $0\in\operatorname{int}\mathcal N$.
In this case, the certified history-model polytope must contain
the true realization and satisfy
\begin{equation}\label{eq:io_channel_certificate}
A_\xi G_\eta=G_\eta A_\nu,\;\mu_{\mathcal N}(A_\nu)\le\gamma^\ast,\;\forall\,(A_\xi,B_\xi)\in\mathcal I_{\mathrm{full}}^{\mathrm{poly},\ast},
\end{equation}
where $A_\nu$ is the induced matrix on the noise coordinates
and $\mu_{\mathcal N}$ is defined as in~\eqref{eq:directed_gain_def}.
Hankel matrices can then be formed from $(u^d,\hat\xi^d)$, and the
implemented feedback, with history-feedback gain $K_\xi$, uses the measured history:
\begin{equation}\label{eq:io_feedback}
u_k=v_{0|k}^\star+K_\xi\hat\xi_k.
\end{equation}

Let $\mathcal Y$ denote the prescribed output-constraint set. Set
$n_\xi:=2n$ and
\begin{equation}\label{eq:io_constraints}
\mathcal X_\xi:=\mathcal U^n\times\mathcal Y^n.
\end{equation}
Define $\mathcal X_{\xi,L}^{\mathrm{feas}}$ as
in~\eqref{eq:X_L_feas_def}, using the history coordinates
and noise bound $\mathcal W_\xi$.
The quadratic costs and terminal-weight selection are
as in Section~\ref{ISS}.

For the structured-noise alternative~\eqref{eq:io_noise_channel}--\eqref{eq:io_channel_certificate},
the full-dimensional noise-set assumption and associated
gain certification are replaced by the conditions stated
there; all other deterministic hypotheses are retained.

\begin{corollary}\label{cor:io_application}
Suppose that the history data and controller ingredients
satisfy the deterministic hypotheses of
Theorems~\ref{thm:rec_feas_shift}--\ref{thm:ISS_value}
under the history formulation and noise-set alternative
specified above.
Then, for every $\xi_0\in\mathcal X_{\xi,L}^{\mathrm{feas}}$,
the controller~\eqref{eq:io_feedback} is recursively
feasible, satisfies $u_k\in\mathcal U$ and
$y_k\in\mathcal Y$ for all $k\in\mathbb N_0$,
and obeys~\eqref{eq:ISS_true_state_bound_value}
with $x$ and $w$ replaced by $\xi$ and $w^\xi$,
respectively.
\end{corollary}

\begin{proof}
Under~\eqref{eq:io_noise_enclosure}, the result follows by
applying the stated theorems in the history coordinates.
For~\eqref{eq:io_noise_channel}, define the gauge on
$\operatorname{im}G_\eta$ by
$p_{\mathcal W_\xi}(G_\eta v):=p_{\mathcal N}(v)$.
Condition~\eqref{eq:io_channel_certificate} gives
\[
-A_\xi\mathcal W_\xi
=G_\eta(-A_\nu\mathcal N)
\subseteq\gamma^\ast\mathcal W_\xi.
\]
Hence
$\eta^+-A_\xi\eta\in(1+\gamma^\ast)\mathcal W_\xi$
for all $\eta,\eta^+\in\mathcal W_\xi$.
The residual bounds and subsequent set-inclusion arguments
therefore remain valid. The remaining feasibility and
stability arguments are unchanged.
Moreover, $u_k$ and $y_k$ are components of $\xi_{k+1}$, so
$\xi_{k+1}\in\mathcal U^n\times\mathcal Y^n$ implies the claimed
input and output constraints. The practical stability estimate follows
with $x$ replaced by $\xi$.
\end{proof}

Establishing less restrictive input--output-specific conditions
for the required data coverage and terminal construction
is left for future work.

\section{Numerical Examples}\label{sec:examples}
\subsection{Input--Output Implementation}\label{sec:io_example}
This example illustrates a history-based input--output implementation
of the proposed design, using measured input and output histories.
We consider a sixth-order three-mass spring--damper system, $\dot x=A_{\rm c}x+B_{\rm c}u$,
$y=Cx$, where $x=\operatorname{col}(q_1,q_2,q_3,\dot q_1,\dot q_2,\dot q_3)$ and
\[
\begin{gathered}
A_{\rm c}=\begin{bmatrix}0_{3\times3}&I_3\\-S&-0.2S\end{bmatrix},\quad
S=\begin{bmatrix}2&-1&0\\-1&2&-1\\0&-1&1\end{bmatrix},\\
B_{\rm c}=\begin{bmatrix}0&0&0&1&0&0\end{bmatrix}^{\!\top},\\
C=\begin{bmatrix}0&0&1&0&0&0\end{bmatrix}.
\end{gathered}
\]
The model is discretized by zero-order hold with $T_{\rm s}=1\,\mathrm{s}$.
The physical model and state are used only for simulation and the
model-based benchmark.

Following Section~\ref{sec:io_implementation}, the controller uses
\[
\hat\xi_k=\operatorname{col}
\bigl(u_{[k-6,k-1]},\hat y_{[k-6,k-1]}\bigr)\in\mathbb R^{12},
\]
where $\hat y_k=y_k+\nu_k$ and $|\nu_k|\le10^{-2}$.
Writing $\hat\xi_k=\xi_k+w_k^\xi$, the known noise channel is
\[
\begin{aligned}
w_k^\xi&=G_\eta\nu_k^{\rm h},\qquad
G_\eta=\begin{bmatrix}0_{6\times6}\\I_6\end{bmatrix},\\
\nu_k^{\rm h}&:=\operatorname{col}(\nu_{k-6},\ldots,\nu_{k-1})\in\mathcal N,\\
\mathcal N&:=10^{-2}[-1,1]^6,\qquad
\mathcal W_\xi=G_\eta\mathcal N.
\end{aligned}
\]
The known history shift relations are incorporated in the model set and
preserve $\operatorname{im}G_\eta$ for every candidate model.
In this noise-channel extension, the interior condition and gauge are
imposed on $\mathcal N\subset\mathbb R^6$, while the tube is constructed
in $\mathbb R^{12}$ using $\mathcal W_\xi$, with
$\epsilon=\max_{w^\xi\in\mathcal W_\xi}\|w^\xi\|_2=\sqrt6\times10^{-2}$.
We compute $K$ and $P_L=c_pP^{-1}$ from the vertex LMI certificate and
construct the RPI error tube using directional noise-image bounds,
finite reachable sums, and certified tail bounds~\cite{WangAngeli2026CDCsub}.

The offline inputs lie in $[-10,10]$, and the offline data have length
$T=1000$. With horizon $L=12$, the offline input is PE
of order $25$, and the stacked history--input regressor has full row rank $13$.
The stage-cost weights are $Q_\xi=I_{12}$ and $R=0.2$, and the constraints
are $|u_k|\le10$ and $|y_k|\le6$. Six zero-input steps from
$x_{-6}=[-4.5,-4.5,-4.5,0,0,0]^\top$ generate the initial history.
The benchmark follows the robust tube-MPC formulation
of~\cite{Mayne2005}, implemented on the true input--output history
realization. It uses the same $(K,P_L,Q_\xi,R,L)$, input--output
information structure, and measurement-noise bound, with its corresponding
disturbance and terminal sets. TRDDPC uses the measured Hankel matrices with
$g\ge0$ and $\mathbf1^\top g=1$.

Figure~\ref{fig:io_comparison} shows the closed-loop output and input
trajectories for both controllers. Both regulate the output to a small
neighbourhood of the origin, while the prescribed input and output
constraints are satisfied throughout the simulation. Both controllers
also produce feasible QPs at every step. As a performance metric, we
consider the sum of the closed-loop stage costs over the $40$ steps,
\[
J=\sum_{k=0}^{39}\bigl(\xi_k^\top Q_\xi\xi_k+Ru_k^2\bigr).
\]
The resulting costs are $734.5528$ for TRDDPC and $734.0812$ for the
model-based robust tube MPC scheme in~\cite{Mayne2005}, corresponding
to a relative increase of only $0.0642\%$. Despite accounting for the
additional multiplicative uncertainty induced by noisy offline data,
the proposed scheme achieves closed-loop performance very close to
that of the model-based benchmark in this example.

\label{par:io_offline_timing}
Offline controller construction, from the saved offline data to the
assembled QP, took $93.86\pm5.91\,\mathrm{s}$ (mean and sample
standard deviation over three executions), using MATLAB R2022b on an
AMD Ryzen 7 5800H computer with 16 GB RAM and one computational thread.
The main contributions were model-polytope preprocessing and vertex
enumeration ($65.62\,\mathrm{s}$), data checks and
$\gamma^\ast$ certification ($14.10\,\mathrm{s}$), and
RPI-tube construction ($3.91\,\mathrm{s}$).

After offline initialization, the average online QP solution times
are $0.303\,\mathrm{s}$ and $0.103\,\mathrm{s}$ per step for TRDDPC
and the model-based MPC benchmark, respectively, both below the sampling
period of $1\,\mathrm{s}$.

The offline burden can nevertheless grow rapidly with state dimension,
since the number of vertices of the
data-consistency polytope can grow exponentially with its dimension in
the worst case.
Outer approximations of the data-consistency set by polytopes with fewer
vertices may reduce this computational burden, at the expense of additional
conservatism~\cite{Kohler2019RAMPC}.

\begin{figure}[t]
  \centering
  \includegraphics[width=0.95\columnwidth]{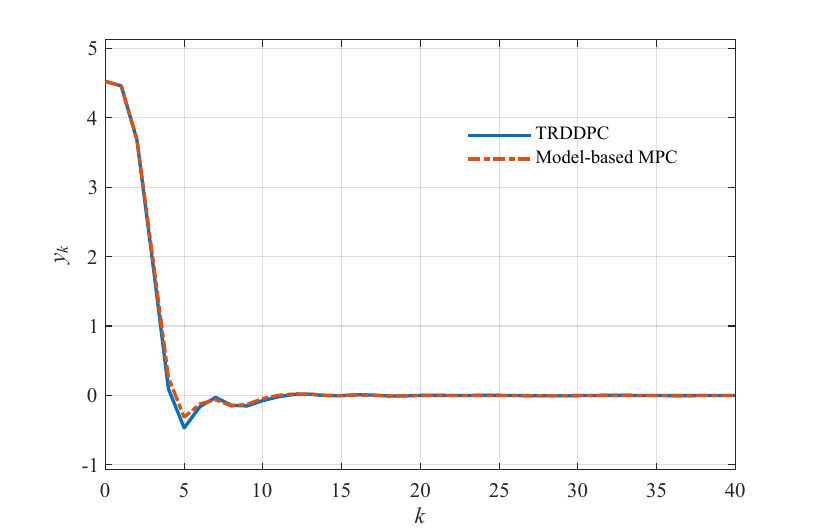}
  \vspace{1mm}

  {\small (a)}

  \vspace{2mm}
  \includegraphics[width=0.95\columnwidth]{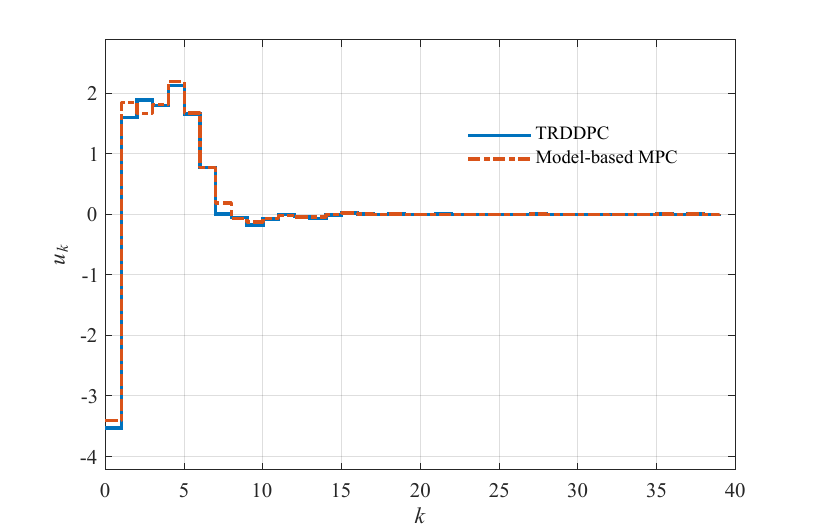}
  \vspace{1mm}

  {\small (b)}
  \caption{Closed-loop responses under TRDDPC and model-based tube MPC:
  (a) output and (b) applied input. Solid blue: TRDDPC; dash-dotted orange:
  model-based tube MPC.}
  \label{fig:io_comparison}
\end{figure}

\subsection{Comparison with Data-Driven MPC Schemes}\label{sec:scalar_example}
To compare the proposed TRDDPC scheme with the robust DeePC scheme
in~\cite{Berberich2020} and the data-driven min--max MPC (DDMMMPC)
scheme in~\cite{xie2026} using the same input and state constraint sets,
we choose a scalar system. In one dimension, the ellipsoidal constraints
in~\cite{xie2026} and the hypercube constraints in~\cite{Berberich2020}
both reduce to interval (polytopic) constraints. We adopt the open-loop
unstable scalar system in~\cite{xie2026},
\[
x_{k+1}=1.1\,x_k+0.5\,u_k,\qquad
\hat x_k = x_k + w_k,
\]
subject to input and state constraints $|u_k|\le 2$ and $|x_k|\le 2$.
The stage cost is given by~\eqref{eq:QR_costs} with $Q=1$ and
$R=0.1$, and we consider $20$ closed-loop steps starting from
$x_0=-1$. An offline input--state trajectory of length $T=100$ is
available. All three schemes use the measurement noise bound
$\|w_k\|_\infty\le 10^{-4}$. Apart from the inherently infinite-horizon formulation
in the DDMMMPC scheme, we employ the same prediction horizon $L=8$ for both
robust DeePC and the proposed TRDDPC.

The closed-loop state and input trajectories for the three schemes
(see Fig.~\ref{fig:scalar_cl}) all converge to a small neighbourhood of
the origin while satisfying the state and input constraints at all
times. We calculated the sum of closed-loop stage
costs over the $20$ steps, yielding
$J_{\mathrm{DDMMMPC}}$, $J_{\mathrm{RDeePC}}$, and $J_{\mathrm{TRDDPC}}$
as reported in Table~\ref{tab:scalar_comparison}.
All three schemes achieve very similar closed-loop performance, with the
proposed TRDDPC attaining the smallest cost: it is approximately
$2.77\%$ lower than that of the DDMMMPC scheme and about
$0.0069\%$ lower than that of the robust DeePC scheme. The average computation times per
iteration in Table~\ref{tab:scalar_comparison} further show that the
TRDDPC QP is about 2.3 times faster than the robust DeePC QP and
roughly four times faster than the SDP-based DDMMMPC.
This behaviour is consistent with the fact that our formulation uses the
smallest set of decision variables and retains a simple quadratic
program.
\begin{figure}[t]
  \centering

  \includegraphics[width=0.9\columnwidth]{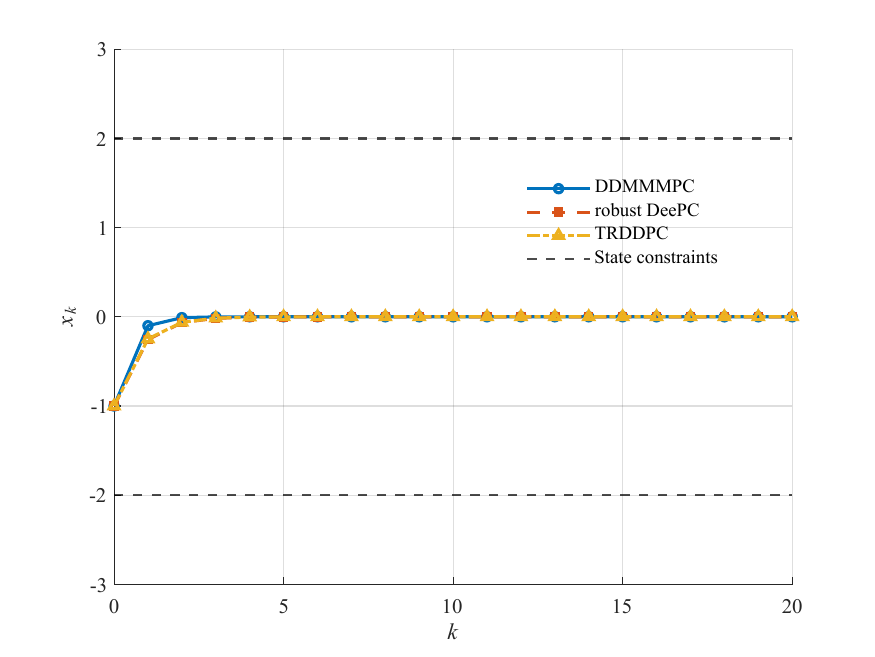}
  \vspace{1mm}

  {\small (a)}

  \vspace{3mm}

  \includegraphics[width=0.9\columnwidth]{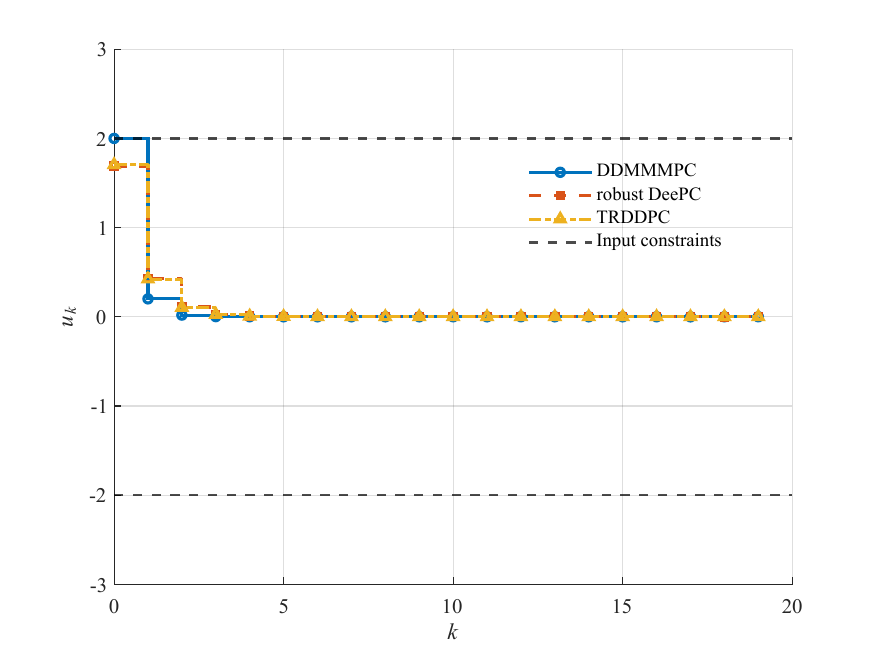}
  \vspace{1mm}

  {\small (b)}

  \caption{Closed-loop state and input trajectories for the scalar system
    under the DDMMMPC of~\cite{xie2026}, the robust
    DeePC scheme of~\cite{Berberich2020}, and the proposed TRDDPC scheme.}
  \label{fig:scalar_cl}
\end{figure}

We also assess the conservatism of each scheme by gradually increasing
the measurement noise bound and checking feasibility. The robust DeePC
scheme becomes infeasible once $\|w_k\|_\infty$ exceeds $1.37357\times
10^{-3}$, while the DDMMMPC scheme
loses feasibility when $\|w_k\|_\infty$ exceeds $0.1288$. In contrast,
the proposed TRDDPC remains feasible and robustly stabilizing for noise
levels up to $\|w_k\|_\infty = 0.188552$. These results indicate that the TRDDPC scheme attains closed-loop
performance slightly better than the DDMMMPC scheme, while more computationally efficient. Compared with the robust DeePC scheme, the proposed controller is significantly less
conservative and still slightly more efficient in terms of computation
time. Moreover, a salient advantage of the proposed approach is that all design
parameters can be computed offline and explicitly from noisy data. In
contrast, in the DDMMMPC scheme some key parameters are only guaranteed
to exist and may not be directly available for offline computation,
whereas robust DeePC requires noise-free data for parameter computation. Finally, to guarantee the sufficient conditions underlying the shift-based recursive-feasibility argument, the proposed
TRDDPC requires a structural offline data condition in addition to the standard persistency-of-excitation condition
commonly assumed in robust DeePC and DDMMMPC.

\begin{table}[t]
  \centering
  \caption{Comparison of data-driven robust MPC schemes}
  \label{tab:scalar_comparison}
  \scriptsize
  \setlength{\tabcolsep}{2pt}
  \begin{tabular}{@{}lcccc@{}}
    \toprule
    & \shortstack{Sum of\\stage costs $J$}
    & \shortstack{Avg.\ comp.\ time\\$\bar t$ [s]}
    & \shortstack{Max.\ admissible\\$\|w\|_\infty$}
    & \shortstack{Design params.\\from noisy data?} \\
    \midrule
    DDMMMPC
      & $1.414120$
      & $1.638\times10^{-2}$
      & $1.288\times10^{-1}$
      & partially\footnotemark[1] \\
    Robust DeePC
      & $1.375020$
      & $9.28\times10^{-3}$
      & $1.37357\times10^{-3}$
      & no\footnotemark[2] \\
    TRDDPC
      & $1.374925$
      & $4.10\times10^{-3}$
      & $1.88552\times10^{-1}$
      & yes \\
    \bottomrule
  \end{tabular}
\end{table}

\footnotetext[1]{Some key parameters are only guaranteed to exist in~\cite{xie2026} and cannot be explicitly computed offline.}
\footnotetext[2]{Tightening parameters in~\cite{Berberich2020}
  are derived under a noise-free data assumption.}

\subsection{Quantifying Conservatism}\label{sec:scalar_conservatism}
To isolate the simplex restriction, we compare noise-free simplex and
model-based predictions with identical horizon, costs, and constraints,
including terminal conditions. We use the first $T$ samples of each of ten
offline trajectories.
Let $d_F^{\rm s}$ and $d_F^{\rm ref}$ denote the widths of the simplex and
reference initial feasible intervals, and $d_\ell^{\rm s}$ and
$d_\ell^{\rm ref}$ the corresponding state-projection widths at $x_0=-1$,
subject to the complete horizon constraints. We define
\[
\rho_F:=\frac{d_F^{\rm s}}{d_F^{\rm ref}},\qquad
\eta:=\min_{1\le\ell\le L}\frac{d_\ell^{\rm s}}{d_\ell^{\rm ref}}.
\]
Table~\ref{tab:scalar_representation_short} reports their medians.
Coverage passes when the Hankel-column convex hull contains every reference
optimal trajectory for $x_0\in[-1,1]$; each check uses one offline trajectory.
All PE tests pass, but one coverage check fails at $T=50$.
At $T=100$, all checks pass although the median $\eta$ is only $87.35\%$.
For this matched problem, convex coverage of the reference optimal family
provides a sufficient data-richness condition for preserving feasibility
and optimal cost on the prescribed operating interval, even when other
admissible trajectories are excluded.
More general a priori data-richness conditions require accounting for the geometry of constrained trajectory families beyond linear spanning and are left for future work.

\begin{table}[t]
\centering
\caption{Matched representation: median width ratios and coverage}
\label{tab:scalar_representation_short}
\footnotesize
\setlength{\tabcolsep}{4pt}
\begin{tabular}{rccc}
\hline
$T$ & $\rho_F$ (\%) & $\eta$ (\%) & Coverage checks passed \\
\hline
50 & 92.85 & 76.03 & 9/10 \\
100 & 100.00 & 87.35 & 10/10 \\
500 & 100.00 & 100.00 & 10/10 \\
\hline
\end{tabular}
\end{table}

Table~\ref{tab:scalar_tightening_short} quantifies tightening and initial
feasibility with noisy offline data ($T=100$).
Here $\mathcal E=[-r_E,r_E]$,
$\tau_x=\operatorname{len}(\mathcal Z)/4$, and
$\tau_u=\operatorname{len}(\hat{\mathcal U})/4$.
The last column uses~\eqref{eq:X_L_feas_def}, requiring initial feasibility
for every admissible measurement error.
At the largest successful sweep point, the state width retains $63.03\%$
but the input width only $0.000966\%$; the next tested value, $0.188560$,
gives an empty tightened input set. Thus the input margin limits this
sweep, quantifying the compatibility requirement after
Assumption~\ref{as:nonempty_tightened}.

\begin{table}[t]
\centering
\caption{Tube size, retained constraint widths, and initial feasibility}
\label{tab:scalar_tightening_short}
\scriptsize
\setlength{\tabcolsep}{2pt}
\begin{tabular}{rcccc}
\hline
$\epsilon$ & $r_E$ & $\tau_x$ (\%) & $\tau_u$ (\%) &
$\mathcal X_L^{\rm feas}\cap\mathcal X$ \\
\hline
$10^{-4}$ & 0.000420 & 99.984 & 99.954 & $[-1.631,1.060]$ \\
0.184 & 0.901413 & 64.129 & 2.776 & $[-1.065,1.023]$ \\
0.188552 & 0.927943 & 63.030 & 0.000966 & $[-0.937,0.920]$ \\
\hline
\end{tabular}
\end{table}

\section{Conclusion}\label{sec:conclusions}
In this paper, we proposed a tractable tube-based robust data-driven predictive control scheme using only finite noisy offline data of an unknown LTI system.
The resulting controller is implemented by solving a convex QP online and provides certified closed-loop guarantees of recursive feasibility, constraint satisfaction, and practical input-to-state stability.
These guarantees are enabled by a simplex-constrained Hankel parametrization and a data-driven tube construction that combines an explicit mismatch bound with certified initial and terminal RPI sets. We also established a conditional input--output extension based on measured histories. Numerical examples show that, compared with existing robust data-driven MPC schemes, the proposed design achieves improved closed-loop performance with less induced conservatism. Our result can also be interpreted as an alternative to tube-based MPC for systems
subject to both parametric uncertainty and measurement noise.

Future work will investigate scaled-simplex parametrizations and their feasibility--performance tradeoffs. Another challenging and important direction is to extend the framework to nonlinear systems.

\appendices
\section{Data verification and acquisition}\label{app:TCC_tail_verify}
\subsection{Verification of Assumption~\ref{as:TCC_tail}}
For the tail convex-coverage condition in
Assumption~\ref{as:TCC_tail}, define the block matrix and target vector
\[
A_{\rm tail}
:=\begin{bmatrix}
H_{u_{[0:L-1]}}\\[.2em]
\hat H_{x_{[0:L-1]}}
\end{bmatrix},
\qquad
b_{\rm tail}
:=\begin{bmatrix}
u^d_{[T-L,T-1]}\\[.2em]
\hat x^d_{[T-L,T-1]}
\end{bmatrix}.
\]
Then Assumption~\ref{as:TCC_tail} holds if and only if the LP
\begin{equation}\label{eq:tail_acceptance_LP}
\text{find }h\in\mathbb R^{T-L}\ \text{s.t. }
A_{\rm tail}h=b_{\rm tail},\ h\ge0,\ \mathbf1^\top h=1
\end{equation}
is feasible. In this case the resulting solution $h$ provides the tail
certificate used in the explicit construction of the shifted coefficient
$g^+$ in Section~\ref{sec:rec-feas}.

\subsection{Sufficient data-acquisition procedure}

We give a sufficient data-acquisition procedure and stochastic conditions
for its almost-sure finite termination. First, a fixed gain $K$ and its
stability certificate $(P,\beta)$ are obtained from noisy pilot data,
for example by applying Proposition~\ref{prop:K-stab} and~\cite{WangAngeli2026CDCsub} to a certified consistency polytope
$\mathcal I_{\rm pre}$ containing $(A^\ast,B^\ast)$.
Conditional on successful certification, Algorithm~\ref{alg:TCC}
collects an offline input--state trajectory under $u_i^d=K\hat x_i^d+v_i^d$,
$v_i^d\in\mathcal V$, until the PE, regression-rank, and tail
convex-coverage tests pass.

For this procedure, the preliminary stability certificate can be
retained by using, throughout the model-set-dependent offline design,
\begin{equation}\label{eq:I_acq}
\mathcal I_{\rm acq}
:=\mathcal I_{\rm pre}\cap\mathcal I_{\mathrm{full}}^{\mathrm{poly},\ast},
\end{equation}
where $\mathcal I_{\mathrm{full}}^{\mathrm{poly},\ast}$ is constructed
from the accepted offline data as in~\eqref{eq:I_full_TAC}.
Since $\mathcal I_{\rm acq}$ contains $(A^\ast,B^\ast)$ and is a
subset of both constituent sets, the certificate for $K$ and all
offline-data residual bounds remain valid. The Hankel predictor
and tail certificate use the accepted offline data; the remaining
invariance and constraint-compatibility conditions are verified
using $\mathcal I_{\rm acq}$.

\begin{algorithm}[t]
\caption{Data acquisition with joint acceptance tests}
\label{alg:TCC}
\begin{algorithmic}[1]
\REQUIRE Fixed pilot-certified $K$, horizon $L$, and deviation set $\mathcal V$.
\STATE Set $s\gets L+n+1$, $T_0\gets(m+1)s-1$, $T\gets0$.
\WHILE{true}
\STATE Measure $\hat x_T^d$, draw $v_T^d\in\mathcal V$, and apply
$u_T^d=K\hat x_T^d+v_T^d$.
\STATE Retain $(u_T^d,\hat x_T^d)$ and set $T\gets T+1$.
\IF{$T\ge T_0$}
\STATE Form $H_s(u^d)$, $Z_0^{\mathrm d}$, $A_{\rm tail}$, and $b_{\rm tail}$
from the offline data.
\IF{$\operatorname{rank}H_s(u^d)=ms$ and $\operatorname{rank}Z_0^{\mathrm d}=n+m$}
\STATE Test the exact LP~\eqref{eq:tail_acceptance_LP}.
\IF{the LP is feasible}
\STATE Return the offline data and their certificate $h$.
\ENDIF
\ENDIF
\ENDIF
\ENDWHILE
\end{algorithmic}
\end{algorithm}

Fix a successful pilot certification and write
$F:=A^\ast+B^\ast K$, which is Schur.
For the acquisition process in Algorithm~\ref{alg:TCC}, define
\[
\zeta_i:=\operatorname{col}(v_i^d,w_i^d),\qquad
T_0:=(m+1)(L+n+1)-1.
\]
Let $\tau$ denote the first data length at which the
algorithm accepts the offline data, with $\tau=\infty$
if the offline data are never accepted.

\begin{proposition}\label{prop:finite_collection}
For the acquisition process above, suppose that
$x_0^d$ takes values in a bounded set and that
$(\zeta_i)_{i\ge0}$ is bounded and i.i.d.,
independent of $x_0^d$, with
$v_i^d\in\mathcal V$ and $w_i^d\in\mathcal W$
almost surely.
Assume that their common law has a Lebesgue density
positive almost everywhere on a neighborhood of the
origin in $\mathbb R^{m+n}$.

Then, almost surely, both rank tests hold for every
$T\ge T_0$, and $\tau<\infty$.
The accepted offline data satisfy
Assumptions~\ref{as:available-trajectory},
\ref{as:fullrowrank}, \ref{as:TCC_tail_V},
and~\ref{as:TCC_tail}.
Moreover, for each fixed pilot certificate and bounded
initial-state set, there exist $M\in\mathbb N$ and $p>0$
such that
\begin{equation}\label{eq:finite_collection_tail}
\mathbb P\{\tau>T_0+jM\mid\text{fixed pilot}\}
\le(1-p)^j,\quad j\in\mathbb N_0.
\end{equation}
\end{proposition}

\begin{proof}
\emph{Rank conditions.}
Let $q_i:=\operatorname{col}(u_i^d,\hat x_i^d)$. Then
\begin{equation}\label{eq:collection_innovation_map}
\begin{aligned}
x_{i+1}^d
&=Fx_i^d+[\,B^\ast\ \ B^\ast K\,]\zeta_i,\\
q_i
&=\begin{bmatrix}K\\I\end{bmatrix}x_i^d
+D_0\zeta_i,\qquad
D_0=\begin{bmatrix}I&K\\0&I\end{bmatrix}.
\end{aligned}
\end{equation}
Since $D_0$ is invertible, the finite-horizon mapping
from innovations to the $q_i$ is affine and invertible,
conditional on the initial state. Thus every finite $q$
sequence has a joint density.

At $T_0$, $H_{L+n+1}(u^d)$ is square, and its determinant
is a nonzero polynomial in the data (e.g., choosing
$u_{j(L+n+1)-1}^d=e_j$ for $j=1,\ldots,m$, with all other
inputs zero, yields a permutation matrix).
Likewise, since $T_0-1\ge n+m$, the full-dimensional
data distribution ensures that an $(n+m)$-minor of $Z_0^{\mathrm d}$
is a nonzero polynomial. Both rank conditions therefore
hold almost surely and remain valid as data are appended.

\emph{Tail convex coverage.}
Schur stability and bounded innovations imply that all
acquisition states are uniformly bounded.
Set $d=L(m+n)$. Each length-$L$ data window satisfies
\[
q_{[a,a+L-1]}
=\mathcal O_Lx_a^d+
\mathcal T_L\zeta_{[a,a+L-1]},
\]
where $\mathcal O_L\in\mathbb R^{d\times n}$ and
$\mathcal T_L\in\mathbb R^{d\times d}$ are the lifted matrices
associated with~\eqref{eq:collection_innovation_map}, and
$\mathcal T_L$ is block lower triangular with invertible diagonal blocks $D_0$.

Choose $d+1$ sufficiently small affinely independent
vectors whose convex hull contains the origin in its
interior, followed by the origin as a final target.
Their corresponding innovation patterns lie in the
neighborhood of positive density.
By inserting sufficiently long near-zero innovation
segments before and between these patterns, Schur stability makes
the starting states uniformly small.
Continuity then ensures that, within suitable open
neighborhoods of these patterns, the final data window
lies in the convex hull of the preceding windows.
These preceding windows are eligible columns of
$A_{\rm tail}$, so the coverage LP is feasible.

The resulting finite innovation pattern has a probability
bounded below by some $p>0$, uniformly over the bounded
acquisition states. Applying it to successive disjoint
blocks of length $M$ and using independence
gives~\eqref{eq:finite_collection_tail}, hence almost-sure
finite termination. Finally,~\eqref{eq:collected_terminal_deviation} ensures
Assumption~\ref{as:TCC_tail_V} for the returned offline data.
\end{proof}

The stochastic conditions in
Proposition~\ref{prop:finite_collection} are sufficient
for almost-sure finite termination.
Any offline data passing the two rank tests and
the LP~\eqref{eq:tail_acceptance_LP} can be used in
the proposed controller, provided the remaining
deterministic design conditions and initial feasibility
hold.

\section{Constants for the stability estimates}\label{app:constants}
For the model family used in the offline design, choose finite bounds
\begin{equation}\label{eq:stability_model_bounds}
\|A+BK\|_2\le\bar F,\quad
\|B\|_2\le\bar B,\quad
\|BK\|_2\le\bar B_K.
\end{equation}
These bounds can be computed as the maxima of the respective matrix
norms over the vertices of the certified polytope.

\subsection{One-step cost constants}\label{app:one_step_constants}
Set $F=A+BK$ and $\delta=Bv+\rho$. Expanding the quadratic costs and
using~\eqref{eq:primal_lyap_final} gives
\begin{equation}\label{eq:complete_cost_expansion}
\begin{aligned}
&J_{\mathrm f}(Fz+\delta)-J_{\mathrm f}(z)
+J_{\mathrm s}(z,Kz+v)\\
&\le-d\|z\|_2^2
+2z^\top\bigl(F^\top P_L\delta+K^\top Rv\bigr)\\
&\qquad+\delta^\top P_L\delta+v^\top Rv.
\end{aligned}
\end{equation}
Apply $-d\|z\|_2^2+2z^\top a\le d^{-1}\|a\|_2^2$ with
$a=F^\top P_L\delta+K^\top Rv$.

Let $p:=\lambda_{\max}(P_L)$, $r:=\lambda_{\max}(R)$, and define
\begin{equation}\label{eq:cost_intermediate_constants}
a_\delta:=p+\frac{2p^2\bar F^2}{d},\qquad
a_v:=r+\frac{2\|K^\top R\|_2^2}{d}.
\end{equation}
The preceding estimate gives
\[
\begin{aligned}
&J_{\mathrm f}(Fz+\delta)-J_{\mathrm f}(z)+J_{\mathrm s}(z,Kz+v)\\
&\le\frac{\|F^\top P_L\delta+K^\top Rv\|_2^2}{d}
+p\|\delta\|_2^2+r\|v\|_2^2\\
&\le a_\delta\|\delta\|_2^2+a_v\|v\|_2^2.
\end{aligned}
\]
Since
\[
\|\delta\|_2^2\le2\bar B^2\|v\|_2^2
+2(1+\gamma^\ast)^2\epsilon^2,
\]
one may take
\begin{equation}\label{eq:cost_constants}
c_w=2a_\delta(1+\gamma^\ast)^2,\qquad
c_v=2a_\delta\bar B^2+a_v.
\end{equation}

\subsection{State-radius constants}
Set
\begin{equation}\label{eq:closed_loop_decay_constants}
M_F:=\sqrt{\frac{\lambda_{\max}(P)}{\lambda_{\min}(P)}},
\qquad\rho_{\mathrm{cl}}:=\sqrt{1-\lambda/2}\in(0,1).
\end{equation}
The feedback certificate implies
$\|F^j\|_2\le M_F\rho_{\mathrm{cl}}^j$ for every certified closed-loop
matrix $F$ and $j\in\mathbb N_0$. Expanding the implemented
state recursion gives
\begin{equation}\label{eq:direct_state_radius_estimate}
\|x_k\|_2\le{}M_F\rho_{\mathrm{cl}}^k\|x_0\|_2 +\frac{M_F(1-\rho_{\mathrm{cl}}^k)}{1-\rho_{\mathrm{cl}}} \bigl(\bar B\nu_{\mathrm d}+\bar B_K\epsilon\bigr).
\end{equation}
Thus the constants in~\eqref{eq:acquisition_radius} may be chosen as
\begin{equation}\label{eq:acquisition_radius_constants}
C_v:=\frac{M_F\bar B}{1-\rho_{\mathrm{cl}}},\qquad
C_w:=\frac{M_F\bar B_K}{1-\rho_{\mathrm{cl}}}.
\end{equation}

\begin{IEEEbiography}[{\includegraphics[width=1in,height=1.25in,clip,keepaspectratio]{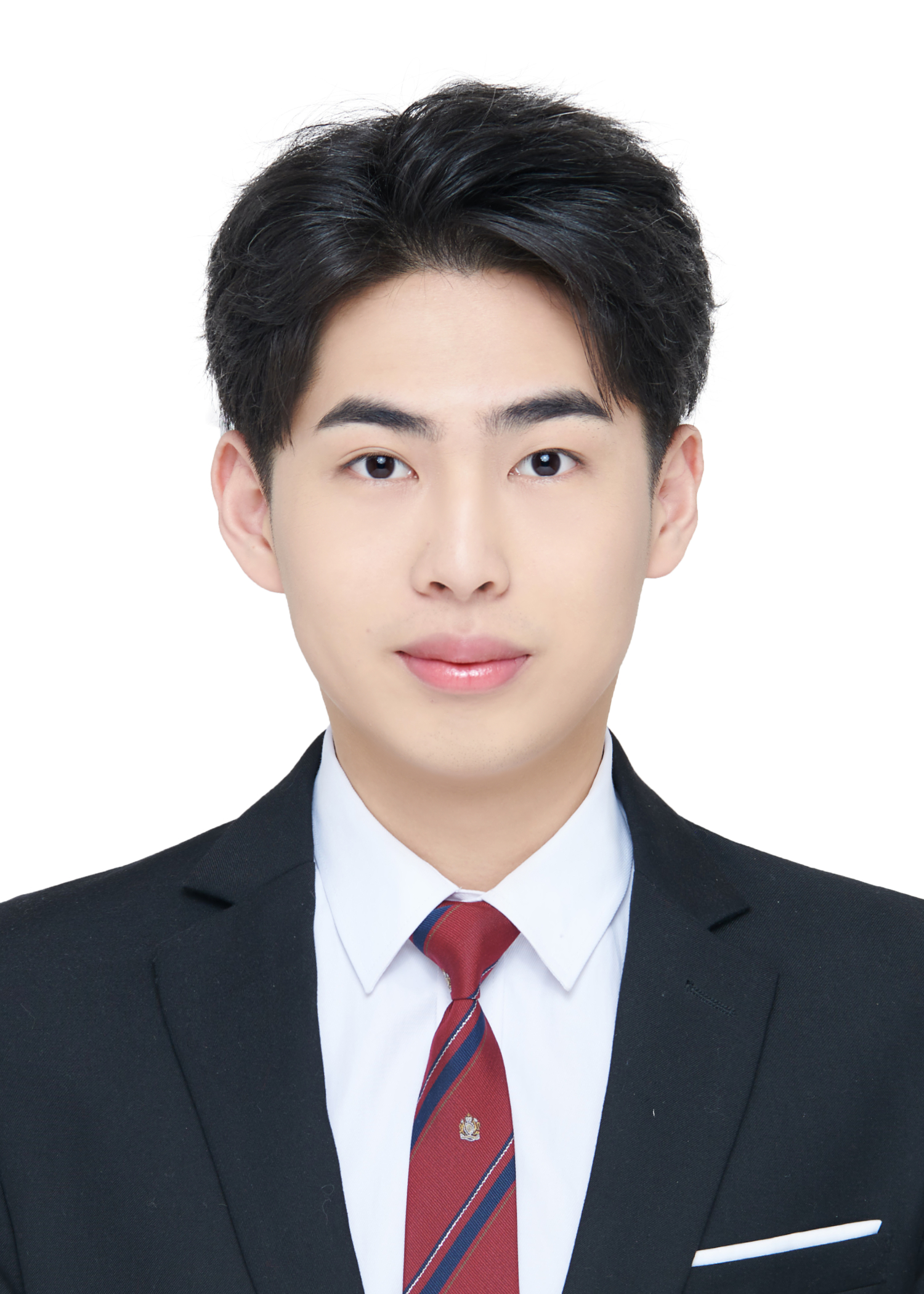}}]{Chi Wang} was born in Zhejiang, China, in 2001. He received the M.Sc. degree in Control and Optimisation from Imperial College London, London, U.K., in 2024. He is currently pursuing the Ph.D. degree with the Control and Power Group, Imperial College London. His current research focuses on the theory of robust data-driven control of constrained systems and predictive control.
\end{IEEEbiography}

\begin{IEEEbiography}[{\includegraphics[width=1in,height=1.25in,clip,keepaspectratio]{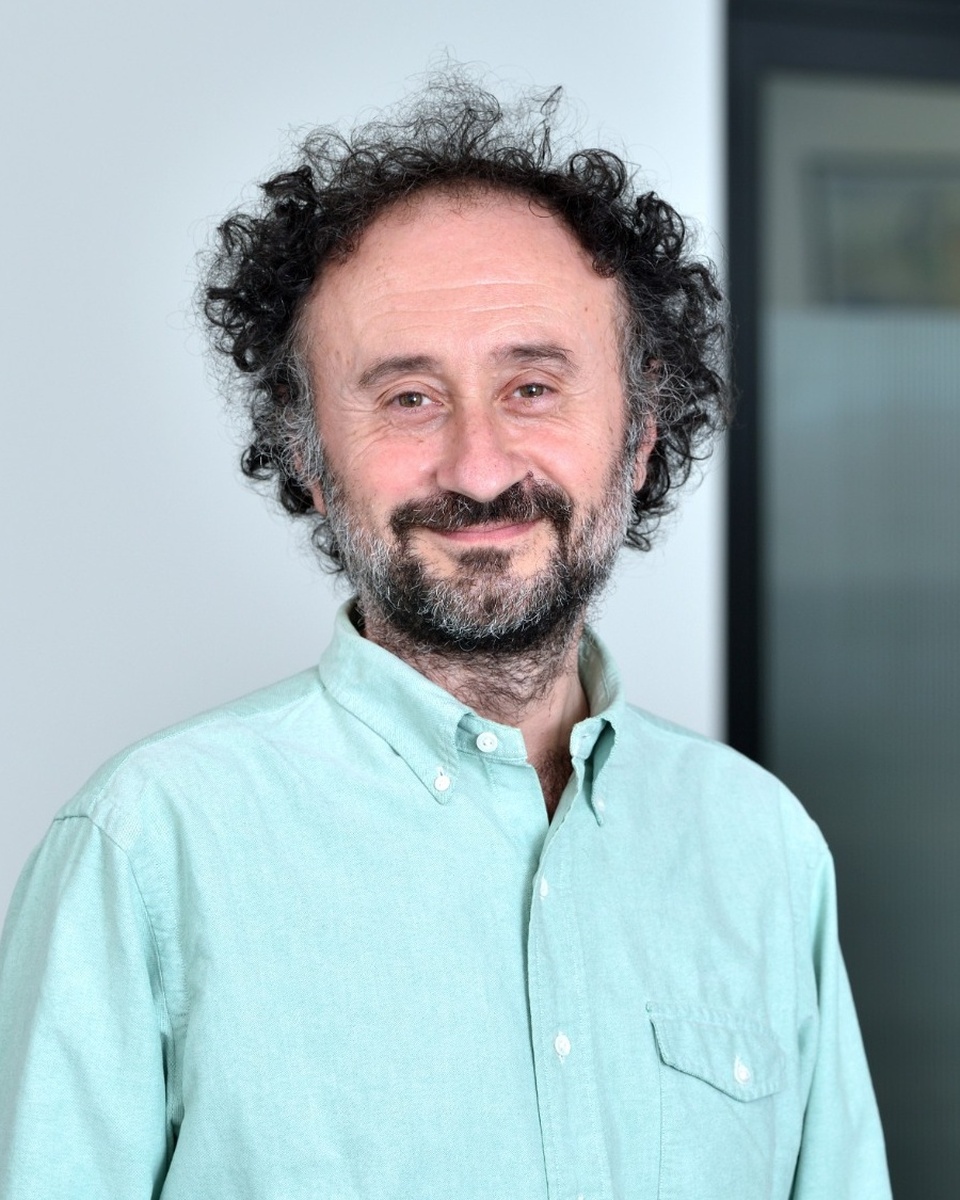}}]{David Angeli} (FIEEE) is a Professor of Nonlinear Network Dynamics with the Department of Electrical and Electronic Engineering, Imperial College London, London, U.K. He received the B.S. degree in computer science engineering and the Ph.D. degree in control theory from the University of Florence, Florence, Italy, in 1996 and 2000, respectively. Since 2000, he has been with the Department of Information Engineering, University of Florence, where he served as an Assistant Professor and, from 2005, as an Associate Professor. In 2007, he was a Visiting Professor with INRIA Rocquencourt, Paris, France. In 2008, he joined the Department of Electrical and Electronic Engineering, Imperial College London, as a Senior Lecturer, where he is currently a Professor and the Director of Postgraduate Teaching. He is the author of more than 120 journal papers in the areas of nonlinear systems stability, control of constrained systems, model predictive control, chemical reaction network theory, and smart grids. Prof. Angeli served as an Associate Editor for the \textit{IEEE Transactions on Automatic Control} and \textit{Automatica}. He was elevated to Fellow of the IEEE in 2015 for contributions to nonlinear control theory. He has been a Fellow of the IET since 2018 and was the recipient of the Honeywell Medal from InstMC in 2021.
\end{IEEEbiography}

\end{document}